\documentclass[a4paper,11pt,onecolumn,unpublished]{quantumarticle}
\pdfoutput=1
\usepackage[T1]{fontenc}
\usepackage[utf8]{inputenc}
\usepackage{amsmath,amssymb,amsthm,mathtools}
\usepackage{microtype,mathrsfs,xcolor,needspace}
\usepackage{algorithm,algpseudocode}
\usepackage[numbers,sort&compress]{natbib}
\usepackage{hyperref}
\hypersetup{
  pdftitle={Sharp universal death of entanglement threshold for Pauli Hamiltonians},
  pdfauthor={Bobak Kiani},
  colorlinks=true,
  linkcolor=blue,
  citecolor=blue,
  urlcolor=cyan
}
\allowdisplaybreaks[2]
\newtheorem{theorem}{Theorem}[section]
\newtheorem{lemma}[theorem]{Lemma}
\newtheorem{proposition}[theorem]{Proposition}
\newtheorem{corollary}[theorem]{Corollary}
\theoremstyle{remark}
\newtheorem{remark}[theorem]{Remark}
\numberwithin{equation}{section}
\DeclareMathOperator{\Tr}{Tr}
\DeclareMathOperator{\supp}{supp}
\DeclareMathOperator{\arctanh}{arctanh}
\DeclareMathOperator{\Sym}{Sym}

\newcommand{\normp}[1]{\left\lVert #1\right\rVert_{\mathrm P,1}}

\newcommand{\comp}{\mathcal C}
\newcommand{\ket}[1]{\lvert #1\rangle}
\newcommand{\bra}[1]{\langle #1\rvert}

\title{Sharp universal death of entanglement threshold for Pauli Hamiltonians}
\author{Bobak T. Kiani}
\email{b.kiani@bowdoin.edu}
\affil{Department of Computer Science, Bowdoin College}
\date{September 9, 2026}
\begin{document}
\maketitle

\begin{abstract}
We determine the exact universal high-temperature separability threshold for Pauli Hamiltonians of bounded degree $\Delta\ge2$. If every coefficient in the Hamiltonian has magnitude at most one and each term has overlapping support with at most $\Delta$ other terms, the Gibbs state is a mixture of product Pauli eigenstates whenever
\[
    \beta \le z_\Delta
    :=\operatorname{arctanh}\left[\max_{0\le x \le 1}x\left(\frac{1-x}{1+x}\right)^{\Delta-1} \right].
\]
For every $\beta>z_\Delta$, a finite commuting Hamiltonian with maximum overlap degree at most $\Delta$ has an entangled Gibbs state. At any fixed $\beta<z_\Delta$ strictly below the threshold, a classical polynomial-time algorithm produces samples from a distribution over product Pauli eigenstates approximating the Gibbs state in trace distance.
\end{abstract}

{\small \noindent
\textbf{Use of artificial intelligence:}
This paper developed out of interactions with ChatGPT 6. I asked the model to find a simple proof of high-temperature separability\footnote{\href{https://chatgpt.com/share/6a9d7033-3b48-83e9-aa68-a8741635eb48}{Here} is a record of the conversation that initiated this study. 
This then split into different threads. The final main theorem was \href{https://chatgpt.com/s/t_6aa41d57c0e08191b998e599ad9265b8}{proven} after a bit of back and forth.}. 
It subsequently proposed the proof strategy and drafted the central arguments that form the basis of this paper. Throughout, I evaluated proofs, pointed out gaps and weaknesses, and requested revisions. A draft of this manuscript was generated with ChatGPT 6 and heavily revised by me. I have checked all proofs and take full responsibility for the final paper. 

\section{Introduction}
\label{sec:intro}

The death of entanglement refers to a phenomenon where a Gibbs state of a quantum system becomes completely unentangled at sufficiently high temperature. States in this completely unentangled regime are fully separable and can be prepared as a mixture of product states.
In a seminal work, Bakshi, Liu, Moitra, and Tang established this phenomenon for local Hamiltonians with bounded interaction degree and gave a classical algorithm for preparing their high-temperature Gibbs states~\cite{BLMT}. We give a different proof of the death of entanglement for Pauli Hamiltonians which determines the exact universal separability threshold at every fixed interaction degree. At bounded degree, this decomposition also admits efficient approximate classical sampling at temperatures above the threshold.
Similar to several classical threshold phenomena~\cite{ScottSokal,Weitz,LSSContraction,ShaoSun,PetersRegtsIsing}, a fixed point of a tree recursion identifies the universal threshold. Here the recursion describes an eigenvalue of the partially transposed Gibbs state of a commuting tree Hamiltonian used in the partial transpose test for entanglement~\cite{Peres,HORODECKI19961}. For every inverse temperature above the threshold, this recursion shows a finite tree Hamiltonian has an entangled Gibbs state at that inverse temperature.

More formally, consider a Hamiltonian on $n$ qubits with a specified Pauli decomposition,
\begin{equation}
 H=\sum_{i\in V}J_iP_i,\qquad
 \rho_\beta(H)=\frac{e^{-\beta H}}{\Tr e^{-\beta H}},\qquad \beta\ge0,
 \label{eq:H}
\end{equation}
where $J_i$ are real coeffients and each $P_i$ is a nonidentity Hermitian Pauli matrix. 
The \emph{term-overlap graph} $G$ has one vertex for each term and contains an edge between two distinct terms when they both act on at least one common qubit. Its maximum degree is at most $\Delta$. We write $N_G[i]$ for the closed neighborhood, including $i$, and $N_G[U]=\bigcup_{i\in U}N_G[i]$.

For $k$-local Pauli Hamiltonians with coefficient magnitudes at most one, \cite{BLMT} established separability when $\beta \le 1/(100k\Delta)$.
Putterman, Zlokapa, and Cotler removed the additional locality factor, proving
separability for $\beta\le1/(96\Delta)$ and treated settings with long-range interactions of bounded total strength at each site~\cite{PZC}. 
They also show that the scaling $O(1/\Delta)$ cannot be improved, but leave open the question about the optimal constant~\cite[Section~1.4]{PZC}.
Both works construct decompositions of Gibbs states algorithmically through successive pinning steps. Their coefficient bounds ensure positivity of the final branches of the decomposition. Our proof instead decomposes the entire expansion of $e^{-\beta H}$ before proving positivity. We express the Gibbs operator as a sum of separable components with real coefficients which are subsequently shown to be positive. The coefficients are partition functions of an independent set polynomial over vertices which are each connected sets of Hamiltonian terms. This connects the proof of separability to the independence polynomial and Shearer's condition for positivity~\cite{Shearer,ScottSokal}.

We now state the precise thresholds. A state is \emph{fully separable} if it is a convex combination of tensor products of single-qubit states. The separable Gibbs states we study here will be convex combinations of product Pauli eigenstates. For $\Delta\ge2$, define
\begin{equation}
\begin{split}    
 z_\Delta
 &:=\operatorname{arctanh}\left[\max_{0\le x \le 1}x\left(\frac{1-x}{1+x}\right)^{\Delta-1} \right]\\
 &=\operatorname{arctanh}\left[x_\Delta\left(\frac{1-x_\Delta}{1+x_\Delta}\right)^{\Delta-1}\right],
 \label{eq:firstzdelta}
\end{split}
\end{equation}
where $x_\Delta=\sqrt{(\Delta-1)^2+1}-\Delta+1$.
Our main result establishes $z_\Delta$ as the universal threshold for separability.

\begin{theorem}[Exact universal separability threshold]
\label{thm:fixed}
Assume Hamiltonian $H$ has coefficients $|J_i|\le J$ and term-overlap graph of maximum degree $\Delta\ge2$. 
Then $\rho_\beta(H)$ is a mixture of product Pauli eigenstates whenever
\begin{equation}
 \beta J\le z_\Delta.
 \label{eq:fixed}
\end{equation}
Conversely, for every $\beta J>z_\Delta$, there is a commuting Pauli Hamiltonian on finitely many qubits with coefficient magnitudes $J$ and maximum overlap degree at most $\Delta$ whose Gibbs state is entangled at that temperature.
\end{theorem}
For example, $z_2=\tfrac14\log2\approx 0.173$ and $z_3\approx 0.0904$. 
The threshold $z_\Delta$ is derived through analyzing fixed points of a recursion. Namely, let $b=\Delta-1$ and define
\begin{equation}
 \begin{split}
 t_\Delta&
 =x_\Delta\left(\frac{1-x_\Delta}{1+x_\Delta}\right)^b,
 \qquad z_\Delta=\arctanh t_\Delta.
 \end{split}
 \label{eq:tDelta}
\end{equation}
The value $t_\Delta$ corresponds to the largest $t$ for which a recursion on a tree $x_{h+1}=t((1+x_h)/(1-x_h))^b$ has a fixed point for $x_h\in(0,1)$.

There is also a useful formulation for weighted Pauli Hamiltonians. Let $\Lambda$ denote the maximum sum of absolute coefficients over a neighborhood of any given term:
\begin{equation}
 \Lambda=\max_{i\in V}\sum_{j\in N_G[i]}|J_j|.
 \label{eq:Lambda}
\end{equation}
\begin{corollary}[Optimal asymptotic constant]
\label{thm:weighted}
For every Pauli Hamiltonian~\eqref{eq:H}, $\rho_\beta(H)$ is a mixture of product Pauli eigenstates whenever
\begin{equation}
 \beta\Lambda\le\frac1{2e}.
 \label{eq:weighted}
\end{equation}
The constant cannot be increased uniformly over Pauli Hamiltonians. For every $c>1/(2e)$, there is a Hamiltonian on finitely many qubits that is entangled at $\beta< c/\Lambda$.
\end{corollary}
Since $\Lambda\le J(\Delta+1)$, a simple sufficient condition is $\beta J(\Delta+1)\le1/(2e)$. For $k$-local Pauli Hamiltonians with per-site strength $B=\max_v\sum_{i:v\in\supp(P_i)}|J_i|$, the inequality $\Lambda\le kB$ gives threshold $\beta kB\le1/(2e)$.

The positive decomposition also gives an efficient approximate sampler
at every fixed relative distance below the threshold.

\begin{theorem}[Classical preparation below the threshold]
\label{thm:sampling}
Fix $\Delta\ge2$ and $\delta\in(0,1)$. Given Hamiltonian $H$, assume $|J_i|\le J$ and the term-overlap graph has
maximum degree at most $\Delta$. For $\beta\ge0$ satisfying
\begin{equation}
 \beta J\le(1-\delta)z_\Delta,
 \label{eq:sampling-temperature}
\end{equation}
and accuracy $\epsilon\in(0,1/2)$, a randomized classical algorithm
outputs samples from a distribution $\mathcal D$ over products of single-qubit Pauli eigenstates $\ket\psi$ such that
\begin{equation}
 \frac12\left\|\mathbb E_{\mathcal D}\,\ket\psi\bra\psi-\rho_\beta(H)\right\|_1
 \le\epsilon.
 \label{eq:sampling-guarantee}
\end{equation}
For fixed $\Delta,\delta$, the running time to produce a sample is at most $\operatorname{poly}(n,m,1/\epsilon)$.\footnote{See Remark~\ref{rem:alg-precision} for dependence on numerical precision.}
\end{theorem}

Section~\ref{sec:sampling} proves this result by sampling connected sets of terms and then product Pauli eigenstates from their separable factors.

\subsection{Proof sketch and relation to earlier work}

To prove separability, one must show the Gibbs state can be decomposed as
\[
\rho_\beta(H)=\sum_{\Gamma} \pi(\Gamma) \rho_\Gamma,
\]
where $\pi(\Gamma)$ is a (classical) distribution over separable components $\rho_\Gamma$. 
Let us rewrite the Hamiltonian to pass signs into Pauli terms so
\[
 -\beta H=\sum_{i\in V}z_iS_i,
 \qquad \text{where } z_i=\beta|J_i|,
 \text{ and } S_i=-\operatorname{sgn}(J_i)P_i.
\]
The starting point for our decomposition is the inclusion--exclusion formula which states that for a function defined on subsets $S\subseteq V$:
\[
 F(V) = \sum_{U\subseteq V} \sum_{B\subseteq U}(-1)^{|U|-|B|}F(B).
\]
Following formulations of quantum cluster expansions~\cite{Park,MH,NF}, we apply this identity with $F(U)=\frac{\exp(\sum_{i\in U}z_iS_i)}{\prod_{i\in U}\cosh z_i}$, thereby decomposing the unnormalized Gibbs state as
\[
\frac{e^{-\beta H}}{\prod_{i\in V}\cosh(z_i)}=\sum_{U\subseteq V}D_U, \quad \text{where } D_U=\sum_{B\subseteq U}(-1)^{|U|-|B|} \,\frac{\exp(\sum_{i\in B}z_iS_i)}{\prod_{i\in B}\cosh z_i}.
\]
Division by $\cosh(z_i)$ factors simplifies later calculations. 
Let $G[U]$ be the term-overlap graph of the Hamiltonian restricted to vertices $U\subseteq V$.
Different connected components of the term-overlap graph $G[U]$ act on disjoint sets of qubits so $D_U$ factorizes as
\[
D_U=\prod_{\gamma\in\comp(G[U])}D_\gamma.
\]
Here $\comp$ denotes the set of connected components of $G[U]$. Each connected component $\gamma$ is denoted as a ``polymer".
As of now $D_\gamma$ may not be separable so we shift each $D_\gamma$ to $-r_\gamma I + (r_\gamma I + D_\gamma)$ where $r_\gamma>0$ is chosen large enough so that $(r_\gamma I + D_\gamma)$ is separable. Decomposing $D_\gamma = \sum_P d_P P$ in the Pauli basis, it suffices for separability to choose $r_\gamma > \sum_P |d_P|$ (larger than the Pauli $1$-norm of $D_\gamma$).
Inserting these shifts and collecting terms, we obtain the key decomposition 
\begin{equation}
\frac{e^{-\beta H}}{\prod_{i\in V}\cosh(z_i)} = \sum_{\Gamma\ \mathrm{compatible}}
 Q_G\!\left(V\setminus N_G\!\left[\bigcup_{\gamma\in\Gamma}\gamma\right]\right)
 \prod_{\gamma\in\Gamma}(r_\gamma I + D_\gamma),
 \label{eq:proofidea_decomp}
\end{equation}
where $Q_G(S_\Gamma)$ is a polynomial equal to 
\[
    Q_G(S_\Gamma)=\sum_{\substack{\Gamma\ \mathrm{compatible}\\\gamma\subseteq S_\Gamma\ (\gamma\in\Gamma)}} \prod_{\gamma\in\Gamma}(-r_\gamma).
\]
Above, $S_\Gamma:=V\setminus N_G\!\left[\bigcup_{\gamma\in\Gamma}\gamma\right]$ is the set of terms in $V$ that are not adjacent to any terms in the polymers $\gamma \in \Gamma$.
Equation~\eqref{eq:proofidea_decomp} is the final decomposition giving a simple criteria for separability:
\[
 Q_G(S)\ge0\ \text{for every }S\subseteq V
 \quad\Longrightarrow\quad
 \rho_\beta(H)\ \text{is fully separable}.
\]
Thus, separability is governed by positivity of a classical independent set polynomial $Q_G(S)$ for which there are well-studied tools to analyze its positivity~\cite{SokalZeros,ScottSokal,Weitz}.
To proceed with the sketch, let us restrict to the commuting setting, where all terms in $H$ mutually commute. Here, 
\[
 D_U=\left(\prod_{i\in U}\tanh z_i\right)\prod_{i\in U}S_i.
\]
For commuting Hamiltonians, $r_\gamma=t_\Delta^{|\gamma|}$ suffices to guarantee separability of $r_\gamma I + D_\gamma$ and the corresponding polynomial is 
\[
 Q_G(S)=q_{G[S]}(t_\Delta),\qquad
 q_G(t_\Delta)=
 \sum_{U\subseteq V(G)}(-1)^{c(G[U])}t_\Delta^{|U|}.
\]
where $c(G[U])$ counts connected components. 
This is a special case (i.e. the case $Q(G;t_\Delta,-1)$) of the subgraph component polynomial previously studied in~\cite{TAM}.
We show that it is positive on every graph of maximum degree at most $\Delta$, giving separability whenever $\beta J\le z_\Delta$.
For noncommuting Hamiltonians, an additional term $r_\gamma = t_\Delta^{|\gamma|}w_A(\gamma)$ is needed where $w_A(\gamma)$ is a correction to control for anticommutation (see Section~\ref{sec:noncommuting}). 

To prove positivity of $q_G(t_\Delta)$, the idea is to inductively remove vertices from $G$ showing
\[
\frac{q_G(t_\Delta)}{q_{G-v}(t_\Delta)} \ge 1-x_\Delta>0 \quad \text{whenever } \operatorname{deg}(v)\le \Delta - 1,
\]
where $G-v$ is the graph with vertex $v$ removed. 
To do this, we can expand $q_G(t_\Delta)$ to obtain
\[
 q_G(t_\Delta)=q_{G-v}(t_\Delta)
 -\sum_{\substack{\gamma\ni v\\G[\gamma]\ \mathrm{connected}}}
          t_\Delta^{|\gamma|}q_{G-N_G[\gamma]}(t_\Delta).
\]
Using the inductive hypothesis,
\[
 1-\frac{q_G(t_\Delta)}{q_{G-v}(t_\Delta)}
 \le\sum_{\substack{\gamma\ni v\\G[\gamma]\ \mathrm{connected}}}
             u^{|\gamma|},
 \qquad u=\frac{t_\Delta}{(1-x_\Delta)^{\Delta-1}}.
\]
It remains to bound this sum of positive weights.
Inspired by~\cite{Weitz,SokalZeros}, we can bound the above sum with an encoding of connected components as paths which are joined in a spanning tree.
The auxiliary tree has a separate vertex for each path, making the recursive count easier.
Setting $C_p$ as the sum of $u^{|U|}$ over connected sets containing a node $p$ within its subtree, then
\[
 C_p=u\prod_{i\ \mathrm{child\ of}\ p}(1+C_i).
\]
Each child is either omitted or contributes one of its connected sets so we find
\[
 u(1+x_\Delta)^{\Delta-1}\le x_\Delta
 \quad\Longleftrightarrow\quad
 t_\Delta\le x_\Delta\left(\frac{1-x_\Delta}{1+x_\Delta}\right)^{\Delta-1}.
\]
This explains the threshold in~\eqref{eq:firstzdelta} which is the required bound for the inductive hypothesis to apply. 
Once this is established, vertices are inductively removed with $q_G(t_\Delta)/q_{G-v}(t_\Delta)>0$ at each step thereby guaranteeing $q_G(t_\Delta)>0$.
Section~\ref{sec:scalar} gives the complete deletion induction.

\begin{figure}[H]
\centering
\begin{tikzpicture}[
  x=1cm,y=1cm,
  term/.style={circle,fill=black,inner sep=2pt},
  qubit/.style={circle,draw,fill=white,minimum size=8mm,inner sep=1pt},
  every node/.style={font=\small},
  every path/.style={line width=0.6pt}
]
\node[font=\small\bfseries] at (0,0.82) {Term-overlap tree};
\node[term,label=above:{$Q_r$}] (r) at (0,0) {};
\node[term,label=left:{$Q_u$}] (u) at (-1.3,-2.30) {};
\draw (u) circle[radius=0.14];
\node[term,label=right:{$Q_v$}] (v) at (1.3,-2.30) {};
\node[term,label=below:{$Q_{u_1}$}] (u1) at (-1.95,-4.60) {};
\node[term,label=below:{$Q_{u_2}$}] (u2) at (-0.65,-4.60) {};
\node[term] (v1) at (0.65,-4.60) {};
\node[term] (v2) at (1.95,-4.60) {};
\draw[->] (r)--node[left,pos=.49] {$e_0$}(u);
\draw[->] (r)--(v);
\draw[->] (u)--node[left,pos=.50] {$e_1$}(u1);
\draw[->] (u)--node[right,pos=.50] {$e_2$}(u2);
\draw[->] (v)--(v1); \draw[->] (v)--(v2);
\node[align=center] at (0,-5.83)
  {$b=2$, $\Delta=3$\\Two physical qubits per edge.};

\node[font=\small\bfseries] at (6.90,0.82) {The six-qubit term $Q_u$};
\node[term,label=above:{$Q_r$}] (rr) at (6.90,0) {};
\node[qubit] (a0) at (6.10,-1.02) {$A_0$};
\node[qubit] (b0) at (7.70,-1.02) {$B_0$};
\node[term,label=left:{$Q_u$}] (uu) at (6.90,-2.30) {};
\draw (uu) circle[radius=0.14];
\draw (rr)--node[left,pos=.45] {$X$}(a0);
\draw (rr)--node[right,pos=.45] {$X$}(b0);
\draw (a0)--node[left,pos=.50] {$Z$}(uu);
\draw (b0)--node[right,pos=.50] {$Z$}(uu);
\draw[dashed] (rr)--(8.10,-0.05);
\node[anchor=west,font=\scriptsize,align=left] at (8.14,-0.05)
  {to the other\\edge};
\node[qubit] (a1) at (4.50,-3.55) {$A_1$};
\node[qubit] (b1) at (5.95,-3.55) {$B_1$};
\node[qubit] (a2) at (7.85,-3.55) {$A_2$};
\node[qubit] (b2) at (9.30,-3.55) {$B_2$};
\draw (uu)--node[above left,pos=.62] {$X$}(a1);
\draw (uu)--node[left,pos=.62] {$X$}(b1);
\draw (uu)--node[right,pos=.62] {$X$}(a2);
\draw (uu)--node[above right,pos=.62] {$X$}(b2);
\node[term,label=below:{$Q_{u_1}$}] (uu1) at (5.225,-4.60) {};
\node[term,label=below:{$Q_{u_2}$}] (uu2) at (8.575,-4.60) {};
\draw (a1)--node[left,pos=.50] {$Z$}(uu1);
\draw (b1)--node[right,pos=.50] {$Z$}(uu1);
\draw (a2)--node[left,pos=.50] {$Z$}(uu2);
\draw (b2)--node[right,pos=.50] {$Z$}(uu2);
\node at (6.90,-5.80)
  {$Q_u=Z_{A_0}Z_{B_0}X_{A_1}X_{B_1}X_{A_2}X_{B_2}$};
\end{tikzpicture}
\caption{A commuting tree Hamiltonian, illustrated for $b=2$. 
Solid vertices represent Hamiltonian terms and there are two qubits, an $A$ and $B$ qubit, per edge (see open circles on the right hand side). The right panel expands the three edges incident to $u$. Incoming and outgoing edges apply $ZZ$ and $XX$ respectively on a vertex. 
The term $Q_u$ therefore has a $ZZ$ applied to the qubit pair on its incoming edge $e_0$, and $XX$ applied to each qubit pair on its outgoing edges $e_1,e_2$.}
\label{fig:proofidea-tree}
\end{figure}
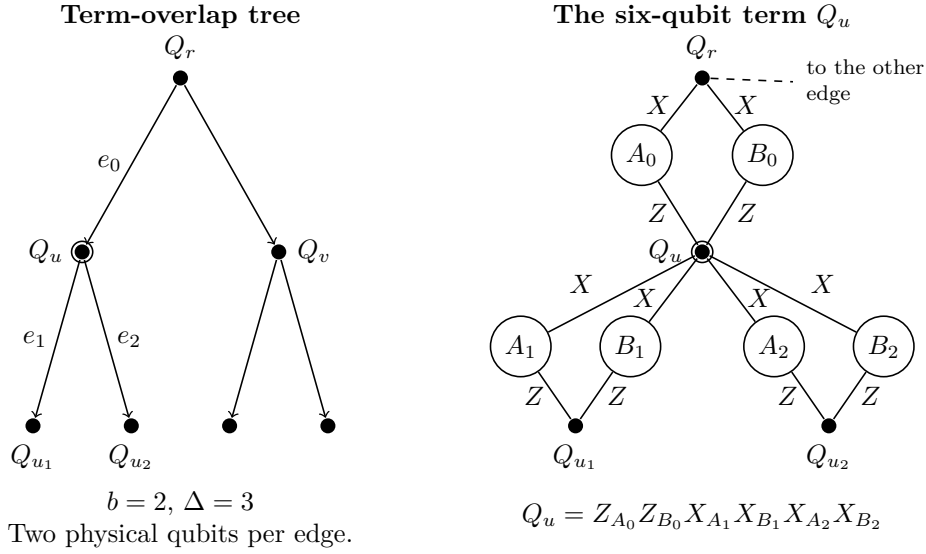

\paragraph{A tree Hamiltonian with entanglement above the threshold}
A Hamiltonian, whose term-overlap graph is a $b$-ary tree as depicted in Figure~\ref{fig:proofidea-tree}, verifies entanglement above the threshold.
Here, $b=\Delta-1$ and there are two qubits $A_e,B_e$ on each edge $e$. The Hamiltonian $H_T$ is
\[
 H_T=-\sum_{v\in V(T)}Q_v,
 \qquad
 Q_v=\prod_{e\ \mathrm{outgoing\ from}\ v}X_{A_e}X_{B_e}
     \prod_{e\ \mathrm{incoming\ to}\ v}Z_{A_e}Z_{B_e}.
\]
All terms commute because Paulis disagree on either zero or two qubits.
The terms are also independent stabilizer generators~\cite{Gottesman}, so their Gibbs state is $\rho_\beta(H_T)=2^{-N}\prod_v(I+sQ_v)$ where $s=\tanh\beta$ and $N=2|E(T)|$.

Using the positive partial transpose test~\cite{Peres,HORODECKI19961} to test for entanglement as in~\cite{Kay}, we transpose all the $A_e$ qubits and determine whether the Gibbs state has a negative eigenvalue thereby confirming it has entanglement. 
Eigenbases of $H_T$ are indexed by signs for each term $Q_v$. 
On a joint eigenvector with every $Q_v=-1$, the partially transposed state has eigenvalue proportional to
\[
 \sum_{U\subseteq V(T)}
 (-1)^{|U|+|E(T[U])|}s^{|U|}
 =q_T(s).
\]
The extra sign $(-1)^{|E(T[U])|}$ comes from the fact that transposes of Pauli $Y^T=-Y$. 
A Pauli $Y$ is on every edge with both endpoints in $U$.
The polynomial above can be recursively computed at height $h$ by splitting it as $q_h=A_h+B_h$, according to whether the root is absent ($A_h$) or present ($B_h$) in $U$.
Joining $b$ child trees gives
\[
 A_{h+1}=(A_h+B_h)^b,
 \qquad B_{h+1}=-s(A_h-B_h)^b.
\]
The ratio $x_h=-B_h/A_h$ thereby lends the recursion
\[
 x_0=s,
 \qquad
 x_{h+1}=s\left(\frac{1+x_h}{1-x_h}\right)^b.
\]
Positive evaluations of $q_h$ imply that $1-x_h$ is also positive.
For $s>t_\Delta$, one finds that $x_h$ eventually exceeds $1$ as this recursion has no fixed point in $(0,1)$ and the iterates continually increase. Thus every $\beta>z_\Delta$ admits a finite tree Hamiltonian with an entangled Gibbs state.
Section~\ref{sec:tree} gives the details.

\paragraph{Sampling the separable state}
Our starting point for sampling is the decomposition $\rho_\beta(H)=\sum_\Gamma\pi(\Gamma)\rho_\Gamma$.
The sampler first samples $\Gamma$ according to $\pi$ and then a product state whose average is $\rho_\Gamma$.
Taking traces in~\eqref{eq:proofidea_decomp} gives both $\pi$ and $\rho_\Gamma$ as
\[
 \begin{aligned}
 \pi(\Gamma)&\propto Q_G(S_\Gamma)
       \prod_{\gamma\in\Gamma}\bigl(r_\gamma+2^{-n}\operatorname{Tr}(D_\gamma)\bigr),\\
 \rho_\Gamma&\propto
       \prod_{\gamma\in\Gamma}(r_\gamma I+D_\gamma).
 \end{aligned}
\]
The harder step is sampling a set of polymers $\Gamma$. We choose polymers one at a time, expressing each conditional probability as a ratio of polymer partition functions. Our positivity bound from earlier and the Scott--Sokal criterion for the independence polynomial~\cite{ScottSokal} gives a zero-free region for these partition functions. We can thus estimate the probabilities using Barvinok's method~\cite{BarvinokBook,HMS} and its efficient implementations for bounded degree models~\cite{PatelRegts,HPR,YYZ} (see also~\cite{PZC,MM,Chen2025Convergence}).
Once $\Gamma$ is chosen, sampling from $\rho_\Gamma$ is simpler. Each factor $r_\gamma I+D_\gamma$ is a combination of terms which are either $I$ or $I\pm P$ for a pauli $P$. In both of these cases, one can sample from the Pauli product states in the support of those projectors. The details are in Section~\ref{sec:sampling}.

\paragraph{Related work} 
Following \cite{MH,NF}, we use inclusion--exclusion to decompose the Gibbs operator into components on connected sets of terms as described in the sketch above. 
For quantum Hamiltonians, the inclusion--exclusion construction originates in Park's work~\cite{Park} and is also studied in Nguyen's thesis~\cite[Chapter~6]{NguyenThesis} and the cluster expansion formulations of Nguyen--Fern\'andez~\cite{NF} and Mann and Minko~\cite{MM}.
Our sampling algorithm leverages zero-freeness to implement Barvinok's algorithm~\cite{BarvinokBook}.
Classical approximation based on zero-freeness for quantum partition functions was developed by Harrow, Mehraban, and Soleimanifar~\cite{HMS} with related polymer implementations in~\cite{MH,MM,PZC,MannWaite}.

The algebraic recursions and identities that prove positivity of the independent set polynomial are similar to those in the classical setting and underlying proofs of Shearer's condition~\cite{Shearer,ScottSokal}. 
Tree recursion also arises in the analysis of thresholds in the classical Ising model~\cite{PetersRegtsIsing,LSSFisher} and algorithms for approximate counting~\cite{Weitz}. 
Our use of these identities is to produce a separable decomposition. The tree comparison follows the method in Sokal~\cite[Proposition~4.2]{SokalZeros} with a choice of spanning tree that also accounts for anticommutation among the components.
In the commuting case, the polynomials in this analysis are a special case of the subgraph component polynomial of Tittmann, Averbouch, and Makowsky~\cite{TAM}. 

Kay’s analysis of thermal graph states~\cite{Kay} on trees uses the partial transpose polynomial that is the basis for our sharpness argument on tree Hamiltonians. Our construction realizes this polynomial with term-overlap graph that is the tree $T$ itself rather than joining terms up to distance two.
The constant $1/(2e)$ in Corollary~\ref{thm:weighted} also occurs in the zero-free analysis of~\cite[Appendix~B]{ZK} when applying the Koteck\'y--Preiss condition~\cite{KP}.

Other related death of entanglement settings have been studied.
Negari, Lessa, and Sahu~\cite{Symmetry} study separability in fixed symmetry sectors including in fermionic systems. Our Gibbs state is on the full Hilbert space in contrast.
\cite{PZC} also study a related ``death of magic'' transition, where at sufficiently high temperature the Gibbs state becomes a mixture of stabilizer states.

\section{Notation and setting}
\label{sec:notation}

Throughout, we consider Pauli Hamiltonians of the form $H=\sum_{i\in V}J_iP_i$ where $J_i \in \mathbb{R}$ such that $|J_i|\le J$, and $P_i$ is an $n$-qubit Pauli. The number of terms is $m=|V|$.
We assume $J_i \neq 0$ for all $i \in V$ as we can always discard
zero coefficients. If $m=0$, the Gibbs state is trivially equal to $I/2^n$.
To simplify the analysis we will at times rewrite Hamiltonians by removing signs from the coefficients and passing them into the Paulis:
\begin{equation}
 -\beta H=\sum_{i\in V}z_iS_i,\qquad
 S_i=-\operatorname{sgn}(J_i)P_i,\qquad z_i=\beta|J_i|\le z.
 \label{eq:reparameterization}
\end{equation}
Above, $z_i > 0$ and pairs $S_i, S_j$ have the same anti-commutation relations as those of $P_i, P_j$. 

The term-overlap graph $G=(V,E)$ has one node for each Pauli term in $H$ and edge set
\[
 E=\bigl\{\{i,j\}:i\ne j,\
       \supp(P_i)\cap\supp(P_j)\ne\varnothing\bigr\}.
\]
We use a degree bound $\Delta\ge2$ satisfying
$\deg_G(i)=|N_G[i]|-1\le\Delta$ for every $i$.
The weighted parameter
$\Lambda=\max_i\sum_{j\in N_G[i]}|J_j|$ controls the strength of terms in the neighborhood of any term.

A related graph is the anticommutation graph.
The anticommutation graph $A=(V,E_A)$ has edge $\{i,j\}\in E_A$ exactly when $P_iP_j=-P_jP_i$ so $E_A\subseteq E$.
For $U\subseteq V$, $G[U]$ denotes the induced graph, $G-U=G[V\setminus U]$,
and $\comp(G[U])$ is its set of connected components. 

An $n$-qubit operator $M \in \mathbb{C}^{2^n \times 2^n}$ is separable if it is a nonnegative sum of product-state projectors, i.e. of the form $M=\sum_{i} q_i \rho_i^{(1)}\otimes\cdots\otimes\rho_i^{(n)}$ where $q_i\ge0$ and $\rho_i^{(1)}, \dots, \rho_i^{(n)}$ are rank $1$ projectors. Dividing any nonzero such operator by its trace gives a fully separable state.
Any operator can be decomposed as a weighted sum of Paulis which gives the so-called Pauli $1$-norm:
\begin{equation}
 M=\sum_P m_PP,\qquad \normp M=\sum_P|m_P|.
 \label{eq:paulinorm}
\end{equation}
Coefficients may be complex unless $M$ is Hermitian. Since a product of Pauli strings is a phase times another Pauli string, the Pauli $1$-norm has the submultiplicativity property
\begin{equation}
 \normp{MN}\le\normp M\normp N.
 \label{eq:paulisubmultiplicativity}
\end{equation}
We denote the normalized trace by $\tau(M)=2^{-n}\Tr M$. 

\paragraph{Threshold parameters.}
Various constants will be used in determining thresholds. We list them here.
At times, we will abbreviate the constants in~\eqref{eq:tDelta} by removing their $\Delta$ subscript:
\begin{equation}
 b=\Delta-1,\quad x=x_\Delta,\quad t=t_\Delta,\quad z=z_\Delta.
 \label{eq:parameters}
\end{equation}
Here $x$ is the critical tree fixed point, $t$ the corresponding weight, and $z=\arctanh t$ the key temperature threshold. 
Note that $2bx=1-x^2$ and  $0<x\le\sqrt2-1<\tfrac12$.
Two further parameters will be introduced later to track cancellations in anticommuting components:
\begin{equation}
 \alpha=1+\frac x2,\qquad \eta=\frac13+\frac x2.
 \label{eq:alphaeta}
\end{equation}
All these quantities depend only on $\Delta$ and not on the temperature.

\section{Proof of separability}
\label{sec:proof}
The proof of separability is given in four steps.
Section~\ref{sec:positive} shows the main decomposition of the Gibbs states into a weighted sum of separable components. The weights are given by the evaluation of a polynomial over connected components of terms in the Hamiltonian. The rest of the proof aims to show these weights are non-negative.
To handle noncommutation, we show that cancellations from noncommuting sets of Paulis helps in bounding contribution from those terms in Section~\ref{sec:noncommuting}. 
Section~\ref{sec:counting} and Section~\ref{sec:scalar} then prove the needed positivity by relating the polynomial on the overlap graph of the Hamiltonian to one defined on trees and showing this tree polynomial is positive.

Throughout, we will use the parameters listed in \eqref{eq:parameters}--\eqref{eq:alphaeta} and the reparameterized form of the Hamiltonian  $-\beta H=\sum_i z_iS_i$ from~\eqref{eq:reparameterization}.
We assume a nontrivial setting with $\beta>0$ and at least one term in the Hamiltonian ($m\ge1$).

\subsection{A decomposition into separable factors}
\label{sec:positive}

We first record a Lemma that when given Hermitian $A$, applies a shift $rI+A$ with $r$ large enough so that the shifted matrix is separable (see also~\cite[Lemma~8]{PZC}).
\begin{lemma}[Separable correction magnitude]
\label{lem:cone}
If $A$ is Hermitian and $r\ge\normp A$, then $rI+A$ is a separable operator composed of product Pauli eigenstate projectors.
\end{lemma}
\begin{proof}
Write $A=\sum_Pa_PP$. Its coefficients are real, and
\[
 rI+A=\left(r-\sum_P|a_P|\right)I
 +\sum_{P:a_P\ne0}|a_P|\bigl(I+\operatorname{sgn}(a_P)P\bigr).
\]
Each $I\pm P$ is positive and diagonal in a product of single-qubit Pauli eigenbases. For $P=I$, it is $2I$ or zero. Every summand is therefore separable.
\end{proof}
Now, we decompose the Gibbs state into separable components.
For each set $U$ of terms, we first divide $\cosh$ factors and define the terms in the inclusion-exclusion argument~\cite{MM,NF}:
\begin{equation}
 F_U=\frac{\exp(\sum_{i\in U}z_iS_i)}{\prod_{i\in U}\cosh z_i},
 \qquad D_U=\sum_{B\subseteq U}(-1)^{|U|-|B|}F_B,
 \qquad F_\varnothing=D_\varnothing=I.
 \label{eq:FD}
\end{equation}
The Gibbs state is then proportional to a sum over $D_U$:
\begin{equation}
 \frac{e^{-\beta H}}{\prod_{i\in V}\cosh(z_i)} = F_V=\sum_{U\subseteq V}D_U.
 \label{eq:inversion}
\end{equation}
To see this, note that
\begin{equation}
    \sum_{U\subseteq V}D_U = \sum_{U\subseteq V}\sum_{B\subseteq U}(-1)^{|U|-|B|}F_B,
\end{equation}
and the coefficient of $F_B$ in the above is $\sum_{U:B\subseteq U\subseteq V}(-1)^{|U|-|B|}$, which is zero unless $B=V$.

Let $G[U]$ be the term-overlap graph of $H$ induced on $U$. Different connected components of $G[U]$ have terms with disjoint physical supports. Their exponentials, normalizing factors, and subset sums in~\eqref{eq:FD} consequently factor. Following the terminology in the statistical physics literature~\cite{NF,MM}, a nonempty connected set $\gamma$ of terms is denoted a \emph{polymer}. Two polymers are \emph{compatible} if their physical supports are disjoint, equivalently if they neither intersect nor are adjacent in $G$. 
Compatible families are precisely the connected components of $G[U]$ as $U$ ranges over subsets of $V$. I.e., the family corresponds to $U=\bigcup_{\gamma\in\Gamma}\gamma$. We have
\begin{equation}
 D_U=\prod_{\gamma\in\comp(G[U])}D_\gamma,
 \qquad
 F_V=\sum_{\Gamma\ \mathrm{compatible}}\prod_{\gamma\in\Gamma}D_\gamma.
 \label{eq:polymer}
\end{equation}
Here $\comp$ denotes connected components, and the empty product is $I$. This is the connected expansion used in~\cite{NF,MM}.

Choose numbers $r_\gamma\ge\normp{D_\gamma}$ (particular values to be chosen later in~\eqref{eq:activities}), and set
\begin{equation}
 R_\gamma=r_\gamma I+D_\gamma,\qquad
 Q_G(S)=\sum_{\substack{\Gamma\ \mathrm{compatible}\\\gamma\subseteq S\ (\gamma\in\Gamma)}}
 \prod_{\gamma\in\Gamma}(-r_\gamma),\qquad Q_G(\varnothing)=1.
 \label{eq:shifts}
\end{equation}
Lemma~\ref{lem:cone} guarantees that $R_\gamma$ is separable. 
The Gibbs state can be decomposed into separable components $R_\gamma$ and coefficient weights $Q_G(S)$ as we show below.
\begin{proposition}[Gibbs decomposition]
\label{prop:positive}
For the choices above,
\begin{equation}
 \frac{e^{-\beta H}}{\prod_{i\in V}\cosh(z_i)}=
 F_V=\sum_{\Gamma\ \mathrm{compatible}}
 Q_G\!\left(V\setminus N_G\!\left[\bigcup_{\gamma\in\Gamma}\gamma\right]\right)
 \prod_{\gamma\in\Gamma}R_\gamma.
 \label{eq:positive}
\end{equation}
If $Q_G(S)\ge0$ for every $S\subseteq V$, then $\rho_\beta(H)$ is a mixture of product Pauli eigenstates.
\end{proposition}
\begin{proof}
Substitute $D_\gamma=R_\gamma-r_\gamma I$ in~\eqref{eq:polymer}, we get
\[
F_V=\sum_{\Gamma\ \mathrm{compatible}}\prod_{\gamma\in\Gamma}(R_\gamma-r_\gamma I).
\]
Expand the product above, a polymer will contribute either as $R_\gamma$ or $-r_\gamma I$. Summing $-r_\gamma$ contributions gives the coefficient in~\eqref{eq:positive}. Within each product $\prod_{\gamma\in\Gamma}R_\gamma$ in~\eqref{eq:positive}, the $R_\gamma$ are separable and have disjoint physical supports, so their product is separable. $\rho_\beta(H)$ is proportional to $F_V$ so $\rho_\beta(H)$ is separable if $Q_G(S) \ge 0$ for every $S\subseteq V$.
\end{proof}
The quantum separability problem has now been reduced to determining whether the coefficients given by the polynomial $Q_G(S)$ for every $S\subseteq V$ are non-negative. The rest of the proof chooses the shifts carefully enough that every one of these coefficients is positive.

\begin{remark}[The commuting case] The proofs that follow greatly simplify in the commuting case.
If the terms in $H$ mutually commute, then
\[
 F_U=\prod_{i\in U}(I+\tanh z_i\,S_i),\qquad
 D_U=\left(\prod_{i\in U}\tanh z_i\right)\prod_{i\in U}S_i.
\]
Thus $r_\gamma=t^{|\gamma|}$ suffices to guarantee $r_\gamma\ge\normp{D_\gamma}$. The corresponding polynomial is
\begin{equation}
 q_G(s)=\sum_{U\subseteq V(G)}(-1)^{c(G[U])}s^{|U|},
 \label{eq:q}
\end{equation}
where $c$ counts connected components, with $c(G[\varnothing])=0$. This is the specialization $Q(G;s,-1)$ of the subgraph component polynomial~\cite{TAM}. For example, a three-vertex path has $q_{P_3}(s)=1-3s-s^2-s^3$. We will prove its positivity at $s=t_\Delta$ as a special case of the same argument that handles noncommuting terms.
\end{remark}

\subsection{Cancellations from anticommuting terms}
\label{sec:noncommuting}

In the commuting setting, one could set $r_\gamma=t^{|\gamma|}$ and proceed to show the polymoial $Q_G(S)\ge 0$ in Proposition~\ref{prop:positive}. 
For noncommuting terms, this nice factorization and choice of $r_\gamma$ is unavailable.
To handle the anti-commuting setting, we will set $r_\gamma=t^{|\gamma|}w_A(\gamma)$ where $w_A(\gamma)$ is a sufficiently bounded correction that is detailed here.

Recall that $A\subseteq G$ is the graph on term labels joining anticommuting pairs. Multiplying a Pauli by a sign does not change this graph. 
For a list of Hermitian Paulis $S_1, \dots, S_M$, allowing repetitions, define its symmetrized product by
\[
 \Sym(S_1,\ldots,S_M)=\frac1{M!}\sum_{\pi\in\mathfrak S_M}S_{\pi(1)}\cdots S_{\pi(M)}.
\]
\begin{lemma}[Cancellation in a symmetrized product]
\label{lem:sym}
If the list of Pauli strings $S_1,\ldots,S_M$ contains at least one anticommuting pair, then
\begin{equation}
 \normp{\Sym(S_1,\ldots,S_M)}\le\frac13.
 \label{eq:symgap}
\end{equation}
\end{lemma}
\begin{proof}
Induct on the number $M$ of occurrences. For $M=2$, the symmetrized product of an anticommuting pair $\Sym(S_1,S_2)=0$. Form the anticommutation graph on $S_1,\ldots,S_M$, treating repeated strings as separate vertices. If its components are $C_1,\ldots,C_r$ with $r>1$, then terms in different components commute, and averaging gives
\[
 \Sym(S_1,\ldots,S_M)=\prod_{j=1}^r\Sym\bigl((S_i)_{i\in C_j}\bigr).
\]
At least one component contains an edge. Apply induction to that component and use the trivial bound $\normp{\Sym}\le1$ for each remaining factor.

If the graph is a star, its leaves commute and its center anticommutes with every leaf. Averaging over the center's position gives a fixed Pauli product times
$M^{-1}\sum_{k=0}^{M-1}(-1)^k$. This is zero for even $M$ and has magnitude at most $1/3$ for odd $M\ge3$.

In a connected graph that is not a star, deleting any vertex leaves an edge. Group permutations by their first occurrence:
\[
 \Sym(S_1,\ldots,S_M)=\frac1M\sum_{i=1}^M
 S_i\,\Sym(S_1,\ldots, S_{i-1}, S_{i+1},\ldots,S_M).
\]
Each remaining list contains an anticommuting pair, so induction and submultiplicativity from \eqref{eq:paulisubmultiplicativity} bound every summand by $1/3$.
\end{proof}

Recall $\alpha=1+x/2$ and $\eta=1/3+x/2$ from~\eqref{eq:alphaeta}.
For a vertex set $U$, define
\begin{equation}
 w_A(U)=\prod_{\substack{C\in\comp(A[U])\\ |C|\ge2}}\eta\,\alpha^{|C|-1}.
 \label{eq:w}
\end{equation}
When the set $\{C\in\comp(A[U]):|C|\ge2\}$ is empty, $w_A(U)=1$. A nontrivial component contributes one factor $\eta$, followed by one factor $\alpha$ for each additional vertex.

\begin{lemma}[Bound on $\normp{D_U}$]
\label{lem:components}
For every $U\subseteq V$ and $\beta J\le z_\Delta$,
\begin{equation}
 \normp{D_U}\le t^{|U|}w_A(U).
 \label{eq:componentbound}
\end{equation}
\end{lemma}
\begin{proof}
The identity $2bx=1-x^2$ and Bernoulli's inequality give
\[
 \left(\frac{1+x}{1-x}\right)^b
 \ge1+\frac{2bx}{1-x}=2+x.
\]
Thus, from~\eqref{eq:tDelta},
\begin{equation}
 t\le\frac{x}{2+x},\qquad
 z=\arctanh t\le\frac12\log(1+x)<\frac x2.
 \label{eq:zbound}
\end{equation}

To bound a set $U$ containing an anticommuting pair, start from the unnormalized sum
\[
 W_U=\sum_{B\subseteq U}(-1)^{|U|-|B|}\exp\!\left(\sum_{i\in B}z_iS_i\right).
\]
Expanding the Taylor series for each $\exp\!\left(\sum_{i\in B}z_iS_i\right)$ above, an ordered Pauli product in the Taylor expansion, with set of labels $L\subseteq U$, is multiplied by
\[
 \sum_{B:L\subseteq B\subseteq U}(-1)^{|U|-|B|}
 =\begin{cases}1,&L=U,\\0,&L\ne U.\end{cases}
\]
Only products using every label survive. If label $i$ appears $m_i\ge1$ times and $M=\sum_i m_i$, there are $M!/\prod_i m_i!$ distinct orderings. Their average is the symmetrized product of the list with those repetitions. Together, this gives
\begin{equation}
 W_U=\sum_{(m_i)_{i\in U}\ge1}
 \left(\prod_{i\in U}\frac{z_i^{m_i}}{m_i!}\right)
 \Sym\bigl(S_i\text{ repeated }m_i\text{ times},\ i\in U\bigr).
 \label{eq:words}
\end{equation}
If $U$ contains an anticommuting pair, every list contains that anticommuting pair since each label appears at least once. Lemma~\ref{lem:sym} gives
\begin{equation}
 \normp{W_U}\le\frac13\prod_{i\in U}(e^{z_i}-1).
 \label{eq:allorders}
\end{equation}
Without an anticommuting pair in $U$, the same argument gives the bound without $1/3$.
Set $c_i=\cosh z_i$ and $h_i=e^{z_i}-1$. Expanding the definitions in~\eqref{eq:FD} yields
\begin{equation}
 D_U=\frac1{\prod_{i\in U}c_i}
 \sum_{B\subseteq U}W_B\prod_{i\in U\setminus B}(1-c_i).
 \label{eq:normalizeW}
\end{equation}
Fix an anticommuting pair $p,q\in U$. Apply~\eqref{eq:allorders} to every $B$ containing both, and the ordinary bound to all other $B$. Since $c_i-1\ge0$, the resulting subset sum is
\begin{equation}
 \normp{D_U}\le\frac1{\prod_{i\in U}c_i}
 \left[\prod_{i\in U}(h_i+c_i-1)
 -\frac23h_ph_q\prod_{i\in U\setminus\{p,q\}}(h_i+c_i-1)\right].
 \label{eq:pairdiscount}
\end{equation}
For $a_i=1+2\tanh(z_i/2)$, the identities
$(h_i+c_i-1)/\sinh z_i=a_i$ and $h_i/\sinh z_i=(a_i+1)/2$ give
\begin{equation}
 \frac{\normp{D_U}}{\prod_{i\in U}\tanh z_i}
 \le\left(\prod_{i\in U}a_i\right)
 \left[1-\frac{(a_p+1)(a_q+1)}{6a_pa_q}\right].
 \label{eq:localai}
\end{equation}
By~\eqref{eq:zbound}, $1\le a_i\le1+z_i\le\alpha$. The bracket is nonnegative and increases with $a_p,a_q$. Hence the last expression is at most
\begin{equation}
 \alpha^{|U|-1}\left(\frac56\alpha-\frac13-\frac1{6\alpha}\right)
 =\alpha^{|U|-1}\left(\frac13+\alpha-1-\frac{(\alpha-1)^2}{6\alpha}\right)
 \le\eta\alpha^{|U|-1}.
 \label{eq:componentcost}
\end{equation}

Finally, different components of $A[U]$ commute with each other. For every $B\subseteq U$, both its exponential and its $\cosh$ factors split across these components, so the inclusion--exclusion sum does too and
\begin{equation}
 D_U=\prod_{C\in\comp(A[U])}D_C.
 \label{eq:redfactor}
\end{equation}
On a singleton $C=\{i\}$, $D_C=\tanh z_i\,S_i$. On each larger component, use~\eqref{eq:componentcost}. Submultiplicativity and $\tanh z_i\le t$ prove~\eqref{eq:componentbound}.
\end{proof}

The factorization in~\eqref{eq:redfactor} is used only to bound a norm: anticommutation components may have overlapping physical supports. In contrast, the compatible factors in Proposition~\ref{prop:positive} have disjoint supports, which is what permits multiplying separable operators.

Returning to~\eqref{eq:shifts}, we can now choose
\begin{equation}
 r_\gamma=t^{|\gamma|}w_A(\gamma).
 \label{eq:activities}
\end{equation}
These activities depend only on the induced graphs on $\gamma$, so they are unchanged in every residual graph containing $\gamma$. Since the weights also multiply over components of $G[U]$, the scalar coefficient from the graph polynomial in Proposition~\ref{prop:positive} can be written as
\begin{equation}
 Q_{G,A}(s)=\sum_{U\subseteq V(G)}(-1)^{c(G[U])}s^{|U|}w_A(U).
 \label{eq:coloredQ}
\end{equation}
The value needed in~\eqref{eq:positive} is $s=t$. Evaluations of the polynomial depend on $G$ and may have either sign so we now show that $Q_{G,A}(t)$ is strictly positive.

\subsection{Counting connected sets by trees}
\label{sec:counting}

The sign in~\eqref{eq:coloredQ} records the number of components of $G[U]$,
while $w_A(U)$ records the sizes of the components of $A[U]$.
The deletion argument in the next subsection separates subsets $U$
according to the connected component $\gamma$ containing a specified
vertex $v$. We therefore need to bound sums over connected sets
$\gamma\ni v$. To do so, we encode each such set by a spanning tree
that preserves its red-component weight.

Color the edges of $A$ red and the other edges of $G$ which are not in $A$ blue. To specify the weight on both graphs and trees, recall $\alpha=1+x/2$ and $\eta=1/3+x/2$ from~\eqref{eq:alphaeta} and denote
\begin{equation}
 g(s)=\begin{cases}1,&s=1,\\ \eta\alpha^{s-1},&s\ge2,\end{cases}
 \qquad
 w_{\rm red}^{F}(U)=\prod_{K\in\mathcal C_{\rm red}^{F}(U)}g(|K|).
 \label{eq:redweight}
\end{equation}
Here $F$ is any graph with red and blue edges, and $\mathcal C_{\rm red}^{F}(U)$ is the set of connected components formed by its red edges with both endpoints in $U$, including isolated vertices. Thus $w_{\rm red}^{G}(U)=w_A(U)$ from~\eqref{eq:w}. 

The counting idea is related to the proof of Sokal's Proposition~4.2~\cite{SokalZeros}. That proposition bounds connected subgraphs and trees through a specified vertex. The proof then encodes a chosen spanning tree by paths in a covering tree. We use that encoding idea, with two changes: simple paths suffice here, and the spanning tree is chosen to preserve every red component. Our encoding takes the form below.

\begin{lemma}[A weight-preserving tree comparison]
\label{lem:lift}
Let $G=(V,E)$ be a finite simple graph, let $A=(V,E_A)$ with subgraph $E_A\subseteq E$ marking the red edges. Fix a root $v=v_0\in V$ (paths will be rooted at $(v_0)$). Define a rooted tree $\mathcal T_v$ as follows. Its vertices are the simple paths
\begin{equation}
 p=(v_0,\ldots,v_k),\qquad v_0,\ldots,v_k\text{ distinct},\quad
 \{v_{i-1},v_i\}\in E\ (1\le i\le k).
 \label{eq:pathtree}
\end{equation}
The root of $\mathcal T_v$ is $(v_0)$. The parent of a path is obtained by deleting its last vertex. The edge from $(v_0,\ldots,v_{k-1})$ to $(v_0,\ldots,v_k)$ has the color of $\{v_{k-1},v_k\}$ in $G$. Define the endpoint map $\pi(p)=v_k$.

For every connected vertex set $\gamma\ni v$, one can choose a connected vertex set $\widehat\gamma\subseteq V(\mathcal T_v)$ containing the root such that $\pi$ restricts to a bijection from $\widehat\gamma$ to $\gamma$ and maps its red components bijectively onto those of $A[\gamma]$. In particular,
\begin{equation}
 \sum_{\substack{\gamma\subseteq V,\ v\in\gamma\\G[\gamma]\ \mathrm{connected}}}
 u^{|\gamma|}w_A(\gamma)
 \le
 \sum_{\substack{U\subseteq V(\mathcal T_v),\ (v_0)\in U\\
                 \mathcal T_v[U]\ \mathrm{connected}}}
 u^{|U|}w_{\rm red}^{\mathcal T_v}(U),
 \qquad u\ge0.
 \label{eq:lift}
\end{equation}
The tree is finite. Its root has $\deg_G(v)$ children, and every other vertex has at most $b=\Delta-1$ children when $G$ has maximum degree at most $\Delta$.
\end{lemma}
\begin{proof}
Fix an order of the edges of $G$ with red edges before the blue edges. For each $\gamma$, start with no edges on its vertex set, and consider all edges of $G[\gamma]$, red edges first and then blue edges, in their fixed orders. Keep an edge precisely when its endpoints are not yet connected by previously kept edges.

After all red edges have been considered, the kept edges connect each component of $A[\gamma]$. Otherwise a red edge joining two remaining pieces would have been kept. They contain no cycle, so within each red component they form a spanning tree. Every subsequent blue edge either joins two current pieces or is omitted. Since $G[\gamma]$ is connected, the final kept graph $S_\gamma$ is a spanning tree. Its red components are exactly those of $A[\gamma]$.

Root $S_\gamma$ at $v_0$. For each $a\in\gamma$, let $p_a$ be its unique path from $v_0$ in $S_\gamma$, and set
\[
 \widehat\gamma=\{p_a:a\in\gamma\}.
\]
These paths are simple, and every prefix of a path is another member of $\widehat\gamma$. Thus they form a rooted subtree of $\mathcal T_v$. The endpoint map is a color-preserving isomorphism from this subtree to $S_\gamma$. It consequently preserves the red-component sizes, giving
\[
 |\widehat\gamma|=|\gamma|,
 \qquad w_{\rm red}^{\mathcal T_v}(\widehat\gamma)=w_A(\gamma).
\]
Also $\pi(\widehat\gamma)=\gamma$, so the encoding is injective. Summing over its images gives the left-hand side of~\eqref{eq:lift}; the other terms on the right are nonnegative.

Every path is connected to the root $(v_0)$, and there are no cycles. Also, every path has length at most $|V|-1$. To verify degree bounds, note that at the root there are $\deg_G(v)\le \Delta$ choices of children in the paths. Elsewhere the parent vertex is unavailable, leaving at most $\Delta-1$ choices.
\end{proof}

With the encoding of connected vertices $\gamma$ into $\mathcal T_v$ above, we now bound the sum $\sum_{\substack{\gamma\ni v\\G[\gamma]\ \mathrm{connected}}}
 u^{|\gamma|}w_A(\gamma)$ with a bound on the corresponding contributions from terms in $\mathcal T_v$.

\begin{lemma}[Bound on the colored tree sum]
\label{lem:rootbound}
With the parameters in~\eqref{eq:parameters} and~\eqref{eq:alphaeta}, set
\begin{equation}
 u=\frac{t}{(1-x)^b}=\frac{x}{(1+x)^b}.
 \label{eq:u}
\end{equation}
For every finite graph $G$ of maximum degree at most $\Delta$, with any fixed red/blue edge coloring,
\begin{equation}
 \sum_{\substack{\gamma\ni v\\G[\gamma]\ \mathrm{connected}}}
 u^{|\gamma|}w_A(\gamma)
 \begin{cases}
 \le x,&\deg_G(v)\le b,\\
 <x(1+x),&\deg_G(v)=b+1.
 \end{cases}
 \label{eq:rootbound}
\end{equation}
\end{lemma}
\begin{proof}
By Lemma~\ref{lem:lift}, let us recall that for $u \ge 0$, we have the bound
\[
\sum_{\substack{\gamma\subseteq V,\ v\in\gamma\\G[\gamma]\ \mathrm{connected}}}
 u^{|\gamma|}w_A(\gamma)
 \le
 \sum_{\substack{U\subseteq V(\mathcal T_v),\ o\in U\\
                 \mathcal T_v[U]\ \mathrm{connected}}}
 u^{|U|}w_{\rm red}^{\mathcal T_v}(U).
\]
The right hand side is defined on a finite rooted colored tree whose vertices apart from the root have at most $b$ children.
We bound this through induction over a recursion starting from a tree with just one root vertex.

\emph{Define two sums in the recursion.}
For a node $p$, let $\mathcal T_p$ consist of $p$ and all its descendants, and let $\mathscr U_p$ be its connected vertex sets containing $p$. For $U\in\mathscr U_p$, denote the red component containing $p$ by $K_p(U)$. Define
\begin{align}
 C_p&=\sum_{U\in\mathscr U_p}
       u^{|U|}\prod_{K\in\mathcal C_{\rm red}^{\mathcal T_p}(U)}g(|K|),
       \label{eq:Cdefinition}\\
 O_p&=\sum_{U\in\mathscr U_p}
       u^{|U|}\alpha^{|K_p(U)|}
       \prod_{\substack{K\in\mathcal C_{\rm red}^{\mathcal T_p}(U)\\K\ne K_p(U)}}g(|K|).
       \label{eq:Odefinition}
\end{align}
Thus $C_p$ is the sum we want. The auxiliary sum $O_p$ replaces the weight of the root's red component by $\alpha^{|K_p(U)|}$. It is the appropriate weight when that component is joined to a selected parent by a red edge.

\emph{Derive the recursions.}
Let $\mathcal B_p$ and $\mathcal R_p$ be the blue and red children of $p$. The same choice of an empty or nonempty subtree at each child gives
\begin{align}
 C_p&=u\prod_{i\in\mathcal B_p}(1+C_i)
 \left[1+\eta\left(\prod_{j\in\mathcal R_p}(1+O_j)-1\right)\right],
 \label{eq:Crec}\\
 O_p&=\alpha u\prod_{i\in\mathcal B_p}(1+C_i)
                   \prod_{j\in\mathcal R_p}(1+O_j).
 \label{eq:Orec}
\end{align}
Let us explain this recursion.
Across any blue edge, the child's red component is separated from that of the root, explaining the $(1+C_i)$ terms. 
Across a selected red edge, the parent and child share a red component and we use $O_j$ to account for this.
In~\eqref{eq:Crec}, the leading $1$ inside the brackets $[1+\eta(\prod_{j\in\mathcal R_p}(1+O_j)-1)]$ accounts for the choice of no red child. 
The product $\prod_{j\in\mathcal R_p}(1+O_j)-1$ then sums all nonempty choices of red children.
Each already starts starts with one red component at $p$, so the entire sum is multiplied by $\eta$ once. In~\eqref{eq:Orec}, the component already extends to the parent, and $p$ contributes $\alpha$ regardless of its children. At a leaf, $C_p=u$ and $O_p=\alpha u$.

For example, a root with two red leaf children has
\[
 C_p=u+2\eta\alpha u^2+\eta\alpha^2u^3.
\]

\emph{Check two scalar inequalities.}
Set $y=3x/2$ and $R=(1+x)/(1+y)$. To show $C_p<x$ and $O_p<y$, we need to establish a few inequalities. The elementary inequality $(1+s)^{-b}\ge1-bs$ for $s\ge0$ follows from convexity. Together with $2bx=1-x^2$, it gives
\begin{equation}
 R^b\ge1-\frac{bx}{2(1+x)}=\frac{3+x}{4}.
 \label{eq:Rbound}
\end{equation}
Since $\alpha=1+x/2$, $\eta=1/3+x/2$, and $0<x<1/2$, we obtain
\begin{equation}
 \frac32R^b-\alpha\ge\frac{1-x}{8}>0,
 \qquad
 \frac23R^b-\eta\ge\frac{1-2x}{6}>0.
 \label{eq:twomargins}
\end{equation}
The first inequality yields
\begin{equation}
 \alpha u(1+y)^b=\frac{\alpha x}{R^b}<y.
 \label{eq:openbound}
\end{equation}
For the other recursion, if $1\le j\le b+1$, then
\begin{equation}
 \frac{(1+x)^j-1}{(1+y)^j-1}
 =\frac23\frac{\sum_{r=0}^{j-1}(1+x)^r}{\sum_{r=0}^{j-1}(1+y)^r}
 \ge\frac23R^{j-1}\ge\frac23R^b>\eta.
 \label{eq:closedcompare}
\end{equation}
The first inequality holds term by term: $(1+x)^r=R^r(1+y)^r\ge R^{j-1}(1+y)^r$. Equivalently,
\[
 1+\eta\bigl((1+y)^j-1\bigr)\le(1+x)^j,
 \qquad 0\le j\le b+1,
\]
with equality only when $j=0$.

\emph{Induct from the leaves.}
We prove $C_p<x$ and $O_p<y$ whenever the subtree rooted at $p$ has at most $b$ children at every node. At a leaf, $C_p=u<x$, and~\eqref{eq:openbound} gives $O_p=\alpha u<y$.

Suppose the bounds hold at the children of $p$. If $p$ has $d\ge1$ children, of which $j$ are red, both recursions are strictly increasing in each child quantity. Thus~\eqref{eq:Crec} and~\eqref{eq:closedcompare} give
\begin{equation}
 C_p<u(1+x)^{d-j}
       \left[1+\eta\bigl((1+y)^j-1\bigr)\right]
 \le u(1+x)^d.
 \label{eq:closedbound}
\end{equation}
For $d\le b$, this is at most $u(1+x)^b=x$. Also, because $C_i<x<y$ and $O_j<y$,
\[
 O_p\le\alpha u(1+y)^d\le\alpha u(1+y)^b<y.
\]
This completes the induction.

At the overall root, all children satisfy these strict bounds even when the root has $b+1$ children. Equation~\eqref{eq:closedbound} then gives
\[
 C_o<u(1+x)^{b+1}=x(1+x).
\]
If it has at most $b$ children, the preceding induction gives $C_o<x$. Applying~\eqref{eq:lift} proves~\eqref{eq:rootbound}. 
\end{proof}

\subsection{Positivity by deleting vertices}
\label{sec:scalar}
We now return to $Q_{G,A}(s)$ from Proposition~\ref{prop:positive} and Equation~\eqref{eq:coloredQ} to establish its positivity for $0\le s \le t_\Delta$. The proof below will establish positivity through an inductive argument starting with the full graph and removing vertices from it one at a time.

\begin{proposition}[Positivity at the tree threshold]
\label{prop:scalar}
For every finite graph $G$ of maximum degree at most $\Delta\ge2$ and every subgraph $A\subseteq G$ on the same vertex set,
\begin{equation}
 Q_{G,A}(s)>0\qquad(0\le s\le t_\Delta).
 \label{eq:scalarpositive}
\end{equation}
In particular, this holds for every induced subgraph, keeping the parameters $\alpha,\eta$ fixed as in~\eqref{eq:alphaeta}.
\end{proposition}
\begin{proof}
The case $s=0$ is immediate. Write $Q_G=Q_{G,A}(s)$, restricting $A$ whenever vertices of $G$ are removed. Induct on the number of vertices, proving positivity together with
\begin{equation}
 \frac{Q_G}{Q_{G-v}}\ge1-x\qquad\text{when }\deg_G(v)\le b.
 \label{eq:ratio}
\end{equation}
The empty graph has value one; a single vertex has $Q_G=1-s\ge1-x>0$.

A selected vertex set either avoids $v$, or has a unique connected component $\gamma$ containing $v$. Summing according to that component gives the deletion recurrence
\begin{equation}
 Q_G=Q_{G-v}
 -\sum_{\substack{\gamma\ni v\\G[\gamma]\ \mathrm{connected}}}
 s^{|\gamma|}w_A(\gamma)Q_{G-N_G[\gamma]}.
 \label{eq:delete}
\end{equation}
This is the usual component decomposition~\cite[p.~106]{TittmannBook}, with our component-dependent weights. All smaller-graph values on the right are positive by induction.

If $\deg_G(v)\le b$, the sum of degrees in $\gamma$ is at most $(b+1)|\gamma|-1$. Its at least $|\gamma|-1$ internal edges are counted twice, leaving at most $(b-1)|\gamma|+1$ external neighbors. Adding the other $|\gamma|-1$ vertices of $\gamma$ gives
\begin{equation}
 |N_G[\gamma]\setminus\{v\}|\le b|\gamma|.
 \label{eq:neighborhood}
\end{equation}
If instead $\deg_G(v)=b+1$, the degree sum is at most $(b+1)|\gamma|$, so there are at most $(b-1)|\gamma|+2$ external neighbors and the bound becomes $b|\gamma|+1$.

Starting from $G-v$, delete the remaining vertices of $\gamma$ along a spanning tree rooted at $v$, always removing a parent before its children. Then delete its external neighbors. Every deleted vertex has already lost a neighbor, so its current degree is at most $b$. Apply~\eqref{eq:ratio} at each deletion. Multiplying the resulting ratios gives
\begin{equation}
 \frac{Q_{G-N_G[\gamma]}}{Q_{G-v}}
 \le\begin{cases}
 (1-x)^{-b|\gamma|},&\deg_G(v)\le b,\\
 (1-x)^{-b|\gamma|-1},&\deg_G(v)=b+1.
 \end{cases}
 \label{eq:ratiodelete}
\end{equation}
For a vertex of degree at most $b$, divide~\eqref{eq:delete} by $Q_{G-v}$, use $s\le t$, and apply Lemma~\ref{lem:rootbound}:
\[
 1-\frac{Q_G}{Q_{G-v}}
 \le\sum_{\gamma\ni v}\left(\frac{t}{(1-x)^b}\right)^{|\gamma|}w_A(\gamma)
 \le x.
\]
This proves~\eqref{eq:ratio}. For a vertex of degree $b+1$, the extra deletion factor in~\eqref{eq:ratiodelete} gives
\begin{equation}
 1-\frac{Q_G}{Q_{G-v}}
 \le\frac1{1-x}\sum_{\gamma\ni v}u^{|\gamma|}w_A(\gamma)
 <\frac{x(1+x)}{1-x}\le1.
 \label{eq:fullroot}
\end{equation}
The last inequality is equivalent to $x\le\sqrt2-1$. The strict inequality, provided by the finite-tree bound, also covers $\Delta=2$, where the last expression equals one. Thus $Q_G>0$ and the induction closes.
\end{proof}

\begin{proof}[Proof of the sufficient condition in Theorem~\ref{thm:fixed}]
Choose the shifts~\eqref{eq:activities}. Lemma~\ref{lem:components} makes each $R_\gamma$ separable. For every $S\subseteq V$, the coefficient in~\eqref{eq:positive} is
$Q_G(S)=Q_{G[S],A[S]}(t)>0$ by Proposition~\ref{prop:scalar}. Proposition~\ref{prop:positive} therefore gives a mixture of product Pauli eigenstates after normalization. No bound on the size of an anticommutation component was used.
\end{proof}

We return to Corollary~\ref{thm:weighted} to prove that $\beta \Lambda \le 1/(2e)$ suffices to guarantee separability. As a reminder, $\Lambda$ is the maximum sum of absolute coefficients over a closed term neighborhood:
\begin{equation}
 \Lambda=\max_{i\in V}\sum_{j\in N_G[i]}|J_j|.
\end{equation}

\begin{proof}[Proof of the sufficient condition in Corollary~\ref{thm:weighted}]
This can be shown by constructing a large degree version of the Hamiltonian with many copies of each term. Suppose $\beta\Lambda<1/(2e)$. For $\varepsilon>0$, replace each term $J_iP_i$ with $N_i=\left\lceil\frac{|J_i|}{\varepsilon}\right\rceil$ many copies each with coefficient $J_i/N_i$.
Every new coefficient has magnitude at most $\varepsilon$, and the largest closed neighborhood size of the split graph is
\[
 D_\varepsilon=\max_i\sum_{j\in N_G[i]}N_j,
 \qquad \varepsilon D_\varepsilon\longrightarrow\Lambda.
\]
Its maximum degree $\Delta_\varepsilon=D_\varepsilon-1$ tends to infinity. For $b=\Delta-1$, the explicit formula~\eqref{eq:tDelta} gives
\begin{equation}
 t_\Delta=\frac{1+O(b^{-2})}{2eb},\qquad
 z_\Delta=\frac{1+O(b^{-2})}{2eb}.
 \label{eq:tasymptotic}
\end{equation}
Indeed, $x_\Delta=(2b)^{-1}+O(b^{-3})$ and
$b\log((1-x_\Delta)/(1+x_\Delta))=-1+O(b^{-2})$.
Thus as $\epsilon \to 0$, we have $\Delta_\varepsilon\beta\varepsilon\to \beta\Lambda$ and $\Delta_\varepsilon z_{\Delta_\varepsilon}\to \frac1{2e}$.

For sufficiently small $\varepsilon$, $\beta\varepsilon<z_{\Delta_\varepsilon}$, and Theorem~\ref{thm:fixed} applies to the unchanged Hamiltonian. At $\beta\Lambda=1/(2e)$, take a limit from below noting that the convex hull of the finitely many product Pauli eigenstate projectors is closed.
\end{proof}

\begin{remark}
    The same asymptotic analysis explains the constant in~\cite[Appendix~B]{ZK} in their cluster expansion bound. Namely, in the maximization defining $t_\Delta$, substitute a trial value $x=\theta/b$. Then, one recovers the optimization in~\cite{ZK} where
    \[
     b\,x\left(\frac{1-x}{1+x}\right)^b\longrightarrow\theta e^{-2\theta}.
    \]
    This is maximized to $1/(2e)$ at $\theta=1/2$. 
\end{remark}

\section{Trees make the bound sharp}
\label{sec:tree}

We now show the sharpness of the transition by realizing a Hamiltonian whose partially transposed thermal graph state has negative eigenvalues, thereby violating the Peres-Horodecki criterion for separable states~\cite{Peres,HORODECKI19961}. Thermal graph states have an extensive history~\cite{HDB,Cavalcanti}. Kay's analysis contains exactly the polynomial we study on trees, as well as a compatible separable decomposition~\cite[Equation~(2) and Section~III.1]{Kay}. Our Hamiltonian has a different placement of physical qubits. The standard graph-state generators $X_v\prod_{u\sim v}Z_u$ have term-overlap graph $T^2$, joining vertices at distance at most two. The construction below has term-overlap graph $T$ itself, which is essential for sharpness in $\Delta$.

Let $T=(V,E)$ be a finite tree with at least two vertices. Put qubits $A_e,B_e$ on each edge, and orient the edges arbitrarily. Define
\begin{equation}
 Q_v=\prod_{e\ \mathrm{outgoing\ from}\ v}X_{A_e}X_{B_e}
     \prod_{e\ \mathrm{incoming\ to}\ v}Z_{A_e}Z_{B_e},
 \qquad H_T=-\sum_{v\in V}Q_v.
 \label{eq:treeH}
\end{equation}
There are $N=2|E|$ physical qubits. The terms of the Hamiltonian above can be mapped to independent checks on a stabilizer code~\cite{Gottesman}.

\begin{lemma}[Gibbs state of tree Hamiltonian]
\label{lem:treegenerators}
The $Q_v$ commute, are independent stabilizer generators, and have term-overlap graph exactly $T$. Every assignment of eigenvalues $\pm1$ occurs, and
\begin{equation}
 Z_{H_T}(w)=2^N(\cosh w)^{|V|},\qquad
 \rho_\beta(H_T)=2^{-N}\prod_v(I+\tanh(\beta) Q_v).
 \label{eq:treeZ}
\end{equation}
The partition-function identity holds for all complex $w$.
\end{lemma}
\begin{proof}
We first show that the terms $Q_V$ are independent stabilizer checks~\cite{Gottesman}. Adjacent terms anticommute on each of their two shared qubits, so they commute overall. Nonadjacent terms have disjoint supports and also commute. Every nonempty product $Q_U=\prod_{v\in U}Q_v$ is not the identity since on an edge meeting $U$, $Q_U$ is either $XX$, $ZZ$, or $YY$ up to sign. Thus no product of terms is a scalar multiple of $I$, giving independent stabilizer checks~\cite{Gottesman}.

Eigenbases of $H_T$ and $\rho_\beta(H_T)$ can be indexed by assignments of signs $s_v\in\{\pm1\}$ to vertices. To see this, note that for any choice of signs $s_v\in\{\pm1\}$, the commuting projector $\prod_v[(I+s_vQ_v)/2]$ has trace $2^{N-|V|}>0$. These projectors are orthogonal for distinct assignments of the signs. Then, expanding each $e^{wQ_v}=\cosh (w)I+\sinh (w)Q_v$ and taking the trace proves~\eqref{eq:treeZ}.
\end{proof}

We use the Peres–Horodecki criterion to test for entanglement~\cite{Peres,HORODECKI19961}. For an operator on $\mathsf A\otimes\mathsf B$, its \emph{partial transpose} transposes only the matrix indices in $\mathsf A$, in the computational basis. Partial transposes of separable states stay positive, so a negative eigenvalue of the partially transposed state certifies the state is entangled~\cite{Peres,HORODECKI19961}. Take $\mathsf A=\{A_e:e\in E\}$ and $\mathsf B=\{B_e:e\in E\}$. A $Y$ occurs on $A_e$ in $Q_U$ exactly when both endpoints of $e$ belong to $U$. Since $Y^T=-Y$ while $X^T=X$ and $Z^T=Z$,
\[
 Q_U^{T_{\mathsf A}}=(-1)^{|E(T[U])|}Q_U.
\]
On a joint eigenvector with $Q_v=-1$, using Lemma~\ref{lem:treegenerators}, the partially transposed state therefore has eigenvalue
\begin{equation}
 2^{-N}\sum_{U\subseteq V}(-1)^{|U|+|E(T[U])|}s^{|U|}
 =2^{-N}q_T(s), \qquad s=\tanh\beta.
 \label{eq:PT}
\end{equation}
The equality uses $|E(T[U])|=|U|-c(T[U])$: every induced subgraph of a tree is a forest. This is exactly Kay's polynomial on a forest~\cite{Kay}.

\Needspace{6\baselineskip}
\begin{lemma}[Negativity above the tree threshold]
\label{lem:treenegative}
For $b=\Delta-1\ge1$ and every $s\in(t_\Delta,1)$, a finite complete rooted $b$-ary tree $T$ has $q_T(s)<0$.
\end{lemma}
\begin{proof}
For a tree of height $h$, write $q_h=A_h+B_h$, where $A_h$ sums the contributions from vertex sets not containing the root and $B_h$ sums those containing it. At height zero,
$A_0=1$ and $B_0=-s$. Joining $b$ child trees gives
\begin{equation}
 A_{h+1}=(A_h+B_h)^b,\qquad
 B_{h+1}=-s(A_h-B_h)^b.
 \label{eq:AB}
\end{equation}
When the root is selected, each selected child component joins its component and its sign flips. While preceding polynomials are positive, $A_h>0$ and setting $x_h=-B_h/A_h$ gives
\begin{equation}
 x_0=s,\qquad x_{h+1}=s\left(\frac{1+x_h}{1-x_h}\right)^b,
 \qquad q_h=A_h(1-x_h).
 \label{eq:cavity}
\end{equation}
A fixed point $y$ of $x_h$ satisfies
\[
s=y\left(\frac{1-y}{1+y}\right)^b.
\]
The right hand side has maximum $t_\Delta$ attained by $y=x_\Delta=\sqrt{b^2+1}-b$.
Thus, when $s>t_\Delta$, there is no fixed point.
Furthermore, if $s>t_\Delta$, the recursion in~\eqref{eq:cavity} forces $x_h$ to iteratively increase. It cannot converge below one, since that would give a fixed point. The iterates cannot approach one without crossing it either, since the map diverges there. Thus some finite iterate reaches or exceeds one. A strict crossing gives $q_h<0$. In the equality case, $A_h=-B_h>0$ and~\eqref{eq:AB} gives $q_{h+1}=-s(2A_h)^b<0$.
\end{proof}

\begin{proof}[Sharpness in Theorem~\ref{thm:fixed} and Corollary~\ref{thm:weighted}]
If $\beta>z_\Delta$, then $s=\tanh\beta>t_\Delta$. Lemma~\ref{lem:treenegative} and~\eqref{eq:PT} give an entangled Gibbs state of overlap degree at most $\Delta$. Multiplying $H_T$ by $J$ proves the claim for general coefficient bound $J$.

For the weighted constant, let $b\to\infty$ and choose $s=(1+1/b)t_\Delta$. The initial iterates in~\eqref{eq:cavity} are $O(1/b)$, so a tree first producing a negative polynomial has height at least two for large $b$. Its maximum degree is $b+1$, and its unit-coefficient Hamiltonian has $\Lambda=b+2$. By~\eqref{eq:tasymptotic},
\[
 \beta\Lambda=(b+2)\arctanh\bigl((1+1/b)t_\Delta\bigr)\longrightarrow\frac1{2e}.
\]
Every example lies strictly above $1/(2e)$ by the separability guarantee already proved. Thus the constant cannot be increased.
\end{proof}

The largest term in this construction acts on $2\Delta$ qubits, although each physical qubit belongs to exactly two terms. 
For $\Delta=2$, the value $t_2=3-2\sqrt2$ also appears as the limiting first positive root of the path recurrence in~\cite[Proposition~16]{TAM} and in the analysis of~\cite[Section~III]{Kay}.

\begin{remark}[Connection to independence polynomial]
For a graph $F$, its independence polynomial is $I_F(\lambda)=\sum_{S\subseteq V(F)\,\mathrm{independent}}\lambda^{|S|}$. At negative $\lambda$, it has a probabilistic interpretation where it enters the inclusion--exclusion conditions for avoiding dependent bad events. This is the setting of Shearer's sharp form of the Lov\'asz local lemma~\cite{Shearer,ScottSokal}. 
Shearer's lemma gives a critical value corresponding to a recursion on a $b$-ary tree of the form $y_{h+1}=p/(1-y_h)^b$ with critical value $\max_{0<y<1}y(1-y)^b=b^b/(b+1)^{b+1}$. Our recursion has the extra numerator $(1+x_h)^b$.
\end{remark}

\begin{remark}[Partition function zeros]
\label{rem:zeros}
The polynomial $q_G$ differs from the physical partition function $Z_H(w)=\Tr e^{-wH}$, and their zero sets need not coincide.
In our tree Hamiltonian, Lemma~\ref{lem:treegenerators} gives $Z_{H_T}(w)=2^N(\cosh w)^{|V|}$. It has no zeros in $|\operatorname{Im}w|<\pi/2$, independently of degree, although its Gibbs state can be entangled at $\beta=O(1/\Delta)$. Thus, zero-freeness alone does not imply separability~\cite{PZC,ZK}.
\end{remark}

\section{Classical preparation below the separability threshold}
\label{sec:sampling}

We prove Theorem~\ref{thm:sampling} giving an efficient algorithm for sampling states in the separable decomposition. The input is an explicit Pauli Hamiltonian $H=\sum_{i\in V}J_iP_i$ on $n$ qubits, with $m=|V|$ terms, $|J_i|\le J$, and term-overlap degree at most $\Delta\ge2$.
The temperature is restricted by $\beta J\le(1-\delta)z_\Delta$.
We construct a randomized classical algorithm that outputs a sample product Pauli eigenstate $\ket\psi$ in time polynomial in $n,m,1/\epsilon$ satisfying
\begin{equation}
 \frac12\left\|\mathbb E\,\ket\psi\bra\psi-\rho_\beta(H)\right\|_1
 \le\epsilon.
 \label{eq:alg-goal}
\end{equation}

Zero coefficients in the Hamiltonian are discarded. If $\beta=0$ or no terms remain, independent uniformly random $Z$ eigenstates give the maximally mixed state exactly. Henceforth assume $\beta>0$ and $m\ge1$. Numerical input and bit-complexity considerations are detailed in Remark~\ref{rem:alg-precision}.

The distribution to be sampled comes from
\eqref{eq:shifts}. Proposition~\ref{prop:positive} shows that in the fully separable regime, the Gibbs state can be decomposed as $\rho_\beta(H)=\sum_{\Gamma}\pi(\Gamma)\rho_\Gamma$ where $\pi$ is a distribution over the compatible families of connected components of terms and $\rho_\Gamma$ is a separable state. 
To be more precise, let us denote
\begin{equation}
 R_\gamma=r_\gamma I+D_\gamma,\qquad
 a_\gamma=\tau(R_\gamma)=r_\gamma+\tau(D_\gamma).
 \label{eq:alg-traces}
\end{equation}
As we are in the separable regime, shifts are chosen so that $a_\gamma$ and $Q_G(S)$ are positive. The contribution of a set of polymers $\Gamma$ is
\begin{equation}
 \pi(\Gamma)=
 \frac{Q_G\!\left(V\setminus N_G[\bigcup_{\gamma\in\Gamma}\gamma]\right)
       \prod_{\gamma\in\Gamma}a_\gamma}{\tau(F_V)}.
 \label{eq:alg-family-law}
\end{equation}
Each separable matrix $R_\gamma$ is nontrivially supported on sites $U_\gamma=\bigcup_{i\in\gamma}\supp(P_i)$ and a compatible set of polymers in $\Gamma$ have support on sites $U_\Gamma=\bigcup_{\gamma\in\Gamma}U_\gamma$. Proposition~\ref{prop:positive} then gives the exact decomposition
\begin{equation}
 \begin{split}
 \rho_\beta(H)&=\sum_{\Gamma\ \mathrm{compatible}}\pi(\Gamma)\rho_\Gamma,\\
 \rho_\Gamma&=
 \left(\bigotimes_{\gamma\in\Gamma}
       \frac{R_\gamma}{\Tr_{U_\gamma}R_\gamma}\right)
       \otimes\frac{I_{U_\Gamma^c}}{2^{|U_\Gamma^c|}}.
 \end{split}
 \label{eq:alg-state-decomposition}
\end{equation}

Section~\ref{sec:alg-distribution} samples $\Gamma$ approximately from
$\pi$ by choosing polymers one at a time. Rather than enumerate all
families in~\eqref{eq:alg-family-law}, we estimate their conditional
weights using polymer partition functions on subsets of the term labels.
To build an $\epsilon$ approximate sampler, only polymers of size $O_\delta(\log(m/\epsilon))$ need be considered as larger choices have exponentially small probability after any previous choices. Section~\ref{sec:alg-oracle} gives the partition-function
estimates, and Section~\ref{sec:alg-states} samples product states from
$\rho_\Gamma$. A polymer of $k$ terms generates at most $2^k$ Pauli strings,
so its mixture terms can be enumerated efficiently at this cutoff, even
when its physical support is large. The decomposition in~\eqref{eq:alg-state-decomposition} is exact, but the cutoff and numerical estimates make the sampler approximate. 

For the estimates we use Barvinok's zero-free Taylor approximation~\cite{BarvinokBook} and bounded-degree coefficient algorithms~\cite{PatelRegts,HPR,YYZ,HMS}. Lemma~\ref{lem:alg-counting} gives a self-contained proof of the required monomial-weight case of~\cite[Theorem~2.2]{HPR}, and we verify its hypotheses for our conditioned models. Related quantum cluster-expansion algorithms appear in~\cite{MH,MM,PZC,MannWaite}.

\subsection{Sampling the polymers}
\label{sec:alg-distribution}

Recall from~\eqref{eq:FD} and~\eqref{eq:inversion} that
\[
 \rho_\beta(H)=\frac{F_V}{\Tr F_V},\qquad
 F_V=\frac{e^{-\beta H}}{\prod_{i\in V}\cosh(\beta|J_i|)}.
\]
For any choice of shifts, the polynomial in~\eqref{eq:shifts} is
\begin{equation}
 Q_G(S)=\sum_{\substack{\Gamma\ \mathrm{compatible}\\
                      \gamma\subseteq S\ (\gamma\in\Gamma)}}
             \prod_{\gamma\in\Gamma}(-r_\gamma),
 \qquad Q_G(\varnothing)=1.
 \label{eq:alg-Q}
\end{equation}
Equation~\eqref{eq:positive} expresses $F_V$ as a sum of products of
$R_\gamma$, with coefficient
$Q_G(V\setminus N_G[\bigcup_{\gamma\in\Gamma}\gamma])$ for family $\Gamma$.
Disjoint supports give
$\tau(\prod_{\gamma\in\Gamma}R_\gamma)=\prod_{\gamma\in\Gamma}a_\gamma$.
Taking normalized traces therefore gives~\eqref{eq:alg-family-law},
including $\sum_\Gamma\pi(\Gamma)=1$.

Shifts chosen for the algorithm here will be slightly different than those used in proving the state is separable. 
Recall shorthand notation $t=t_\Delta$, $z=z_\Delta$, and the weights
$w_A$ from~\eqref{eq:parameters} and~\eqref{eq:w}. Define
\begin{equation}
 \sigma=\frac{t+\tanh(\beta J)}2,\qquad
 \mathcal R=\frac{t}{\sigma}>1,\qquad
 r_\gamma=\sigma^{|\gamma|}w_A(\gamma).
 \label{eq:alg-parameters}
\end{equation}
The intermediate value $\sigma$ allows for a margin for a zero-free disk of radius $\mathcal R$. These shifts replace~\eqref{eq:activities} only in this section. Keeping the individual $\tanh z_i$ factors in the proof of Lemma~\ref{lem:components}, specifically \eqref{eq:localai}--\eqref{eq:redfactor}, gives
\begin{equation}
 \normp{D_\gamma}\le\tanh(\beta J)^{|\gamma|}w_A(\gamma)
 \le(2-\mathcal R)r_\gamma,
 \qquad r_\gamma-\normp{D_\gamma}\ge(\mathcal R-1)r_\gamma>0.
 \label{eq:alg-cone-margin}
\end{equation}
Here $\tanh(\beta J)/\sigma=2-\mathcal R$.
Thus $0<a_\gamma\le2r_\gamma$, and Proposition~\ref{prop:scalar} ensures that $Q_{G[S],A[S]}(\sigma)>0$ as required for separability. 

We now sample a family of polymers from the distribution $\pi$ one step at a time. Each step will either declare that a term $v$ belongs to no polymer or a polymer containing $v$ is in the family.
To sample, first choose an undecided term $v$. 
Either we declare that $v$ belongs to no selected polymer, with probability $p_\varnothing$, or we select a polymer $\gamma$ containing $v$, with probability $p_\gamma$.
These probabilites are given by a polymer partition function which we assume oracle access to for now. Section~\ref{sec:alg-oracle} later justifies the assumption by giving an efficient algorithm for estimating the partition functions.

At each step, we track subset $S\subseteq V$ storing the labels not blocked by polymers already selected, and $B\subseteq S$ storing the labels declared not to belong to any selected polymer. The undecided labels are $S\setminus B$.
Selecting $\gamma$ removes its closed neighborhood from $S$.
Declaring $v$ unselected adds it to $B$ but does not remove it from $S$.
The partition function used to sample is given by the total weight of the allowed families in~\eqref{eq:alg-family-law} conditioned on subsets $S$ and $B$:
\begin{equation}
 \mathcal Z(S,B)=
 \sum_{\substack{\Gamma\ \mathrm{compatible}\\
                  \gamma\subseteq S\setminus B\ (\gamma\in\Gamma)}}
 Q_G\!\left(S\setminus N_G[\bigcup_{\gamma\in\Gamma}\gamma]\right)
 \prod_{\gamma\in\Gamma}a_\gamma>0.
 \label{eq:alg-mixed-positive}
\end{equation}
This is a partition function measuring the weight of allowed completions conditioned on previous polymers selected and vertices declared not to be selected. The factors $a_\gamma$ from previously selected polymers are common to every completion and cancel when conditioning. Initially $\mathcal Z(V,\varnothing)=\tau(F_V)$.

For an eligible polymer $\gamma$, define the sets remaining after its selection by
\begin{equation}
 S_\gamma=S\setminus N_G[\gamma],\qquad
 B_\gamma=B\cap S_\gamma.
 \label{eq:alg-child-sets}
\end{equation}
For chosen $v\in S\setminus B$,
either $v$ is placed into $B$ or a polymer $\gamma$ containing $v$ is selected so the sum in~\eqref{eq:alg-mixed-positive}
splits into
\begin{equation}
 \mathcal Z(S,B)=\mathcal Z(S,B\cup\{v\})+
 \sum_{\substack{v\in\gamma\subseteq S\setminus B\\
                         G[\gamma]\ \mathrm{connected}}}
      a_\gamma\mathcal Z(S_\gamma,B_\gamma).
 \label{eq:alg-decision-sum}
\end{equation}
Conditional probabilities for not selecting $v$ or selecting a polymer $\gamma$ are consequently
\begin{equation}
 p_\varnothing=\frac{\mathcal Z(S,B\cup\{v\})}{\mathcal Z(S,B)},
 \qquad
 p_\gamma=a_\gamma\frac{\mathcal Z(S_\gamma,B_\gamma)}{\mathcal Z(S,B)}.
 \label{eq:alg-pgamma}
\end{equation}
For example, for a single Pauli term, $\tau(D_{\{v\}})=0$, so these probabilities are $p_\varnothing=1-r_{\{v\}}$ and $p_{\{v\}}=r_{\{v\}}$.
Starting at $(S,B)=(V,\varnothing)$ and stopping when $S=B$ samples from the distribution~\eqref{eq:alg-family-law} in at most $m$ decisions. 
We are helped by the fact that polymers of large size have a vanishing contribution to $\mathcal Z(S,B)$. 
\begin{lemma}[Conditional polymer-size bound]
\label{lem:alg-tail}
For every $B\subseteq S\subseteq V$, $v\in S\setminus B$, and integer
$K\ge0$,
\begin{equation}
 \sum_{\substack{v\in\gamma\subseteq S\setminus B\\
             G[\gamma]\ \mathrm{connected},\ |\gamma|>K}}p_\gamma
 \le\frac{2}{\mathcal R^{K+1}-1}.
 \label{eq:alg-tail}
\end{equation}
\end{lemma}
The proof is in Section~\ref{sec:alg-oracle}. 
We will use a partition function estimator $\operatorname{EstimateZ}(S,B,\xi)$ that returns an estimate of $\mathcal Z(S,B)$ to within a factor $e^{\pm\xi}$ error (see Section~\ref{sec:alg-oracle}).
We will set cutoff and error to
\begin{equation}
 K=\left\lceil\frac4\delta\log\left(1+\frac{16m}{\epsilon}\right)\right\rceil,
 \qquad \xi=\frac{\epsilon}{100m}.
 \label{eq:alg-cutoff}
\end{equation}
Let $\mathcal C$ be the eligible polymers of size at most $K$. Our sampler will estimate weights
\begin{equation}
 \widehat W_\varnothing=\operatorname{EstimateZ}(S,B\cup\{v\},\xi),
 \qquad
 \widehat W_\gamma=\widehat a_\gamma
                   \operatorname{EstimateZ}(S_\gamma,B_\gamma,\xi)
 \label{eq:alg-continuation-weights}
\end{equation}
and draw samples according to
\begin{equation}
 \Pr(Y=y)=\frac{\widehat W_y}
 {\widehat W_\varnothing+\sum_{\gamma\in\mathcal C}\widehat W_\gamma},
 \qquad y\in\{\varnothing\}\cup\mathcal C.
 \label{eq:alg-draw}
\end{equation}
With exact weights, this is the exact decision law conditioned on not selecting a polymer larger than $K$. Its total variation distance from the full decision law is the omitted probability, bounded by Lemma~\ref{lem:alg-tail}. 

There are $(\Delta+1)^{O(K)}$ candidate polymers containing $v$, which is enumerable in polynomial time~\cite[Lemmas~2.4--2.5]{HPR}.
The oracle estimates the full partition function
$\mathcal Z(S,B)$ as the value $K$ restricts only the polymers that can be selected. The full algorithm to sample polymers is given below in Algorithm~\ref{alg:product-sampler}.

\begin{algorithm}[H]
\caption{$\operatorname{SamplePolymers}$: approximate sampling of~\eqref{eq:alg-family-law}}
\label{alg:product-sampler}
\small
\begin{algorithmic}[1]
\Require The preprocessed instance and $K,\xi$ from~\eqref{eq:alg-cutoff}.
\State $S\gets V$; $B\gets\varnothing$; $\Gamma\gets\varnothing$.
\While{$S\setminus B\ne\varnothing$}
 \State Let $v$ be the smallest label in $S\setminus B$.
 \State List $\mathcal C=\{\gamma\subseteq S\setminus B:\ v\in\gamma,\ |\gamma|\le K,
                                     \ G[\gamma]\text{ connected}\}$.
 \State $\widehat W_\varnothing\gets\operatorname{EstimateZ}(S,B\cup\{v\},\xi)$.
 \For{$\gamma\in\mathcal C$}
  \State Compute $\widehat a_\gamma>0$ to relative error $\xi$.
  \State $S_\gamma\gets S\setminus N_G[\gamma]$; $B_\gamma\gets B\cap S_\gamma$.
  \State $\widehat W_\gamma\gets\widehat a_\gamma\operatorname{EstimateZ}(S_\gamma,B_\gamma,\xi)$.
 \EndFor
 \State Draw $Y\in\{\varnothing\}\cup\mathcal C$ with probabilities~\eqref{eq:alg-draw}.
 \If{$Y=\varnothing$}
  \State $B\gets B\cup\{v\}$.
 \Else
  \State $\gamma\gets Y$; $\Gamma\gets\Gamma\cup\{\gamma\}$.
  \State $S\gets S\setminus N_G[\gamma]$; then $B\gets B\cap S$ using the updated $S$.
 \EndIf
\EndWhile
\State \Return $\Gamma$.
\end{algorithmic}
\end{algorithm}

\subsection{A partition function oracle from zero-freeness}
\label{sec:alg-oracle}
We will use Barvinok's method to estimate the partition function $\mathcal Z(S,B)$~\cite{HPR,BarvinokBook}.
To proceed, let us rewrite the partition function~\eqref{eq:alg-mixed-positive} in a convenient form. For polymer weights $v_\gamma$ define
\begin{equation}
 \Xi(\mathbf v)=\sum_{\Gamma\ \mathrm{compatible}}
                         \prod_{\gamma\in\Gamma}v_\gamma.
 \label{eq:alg-Xi}
\end{equation}
The partition function of interest is $\mathcal Z(S,B)=\Xi(\mathbf v^{S,B})$, where
\begin{equation}
 v_\gamma^{S,B}=
 \begin{cases}
  0,&\gamma\not\subseteq S
       \quad\text{(unavailable)},\\
  -r_\gamma,&\gamma\subseteq S,\ \gamma\cap B\ne\varnothing
       \quad\text{(cannot be selected)},\\
  \tau(D_\gamma),&\gamma\subseteq S\setminus B
       \quad\text{(selection undecided)}.
 \end{cases}
 \label{eq:alg-mixed}
\end{equation}
To see why, note that an undecided polymer has contributions from both terms in $\tau(D_\gamma)=-r_\gamma+a_\gamma$ where $-r_\gamma$ and $a_\gamma$ are the weights given if the polymer is not selected or selected respectively. A polymer meeting $B$ cannot be selected so it supplies a contribution $-r_\gamma$. 
Expand these two choices for every undecided polymer in~\eqref{eq:alg-Xi} and collect the selected $a_\gamma$ factors, exactly as in Proposition~\ref{prop:positive}. This gives $\mathcal Z(S,B)$ as required.

Barvinok's method~\cite{BarvinokBook} approximates $\log f_{S,B}(1)=\log\mathcal Z(S,B)$ via its truncated Taylor series in $q$ about zero. An additive error $\xi$ approximation to $\log  Z(S,B)$ gives a multiplicative factor approximation of $e^{\pm\xi}$ to the partition function $ Z(S,B)$. 
The following proposition states the final error guarantee which we prove at the end of this subsection.
\begin{proposition}[Partition function oracle]
\label{prop:alg-estimateZ}
For every $B\subseteq S\subseteq V$ and $0<\xi<1$, a deterministic
algorithm $\operatorname{EstimateZ}(S,B,\xi)$ returns $\widehat Z>0$ with
\begin{equation}
 e^{-\xi}\le\frac{\widehat Z}{\mathcal Z(S,B)}\le e^\xi.
 \label{eq:alg-Z-contract}
\end{equation}
The algorithm truncates the Taylor series of $\log f_{S,B}(q)$ about $q=0$ to order
\begin{equation}
 L=\left\lceil\frac4\delta\log\frac{16(m+1)}{\delta\xi}\right\rceil.
 \label{eq:alg-oracle-cost}
\end{equation}
The runtime is $\operatorname{poly}(n,m,L)(\Delta+1)^{O(L)}$. The order $L$ bounds the Taylor series truncation error by $\xi/4$ (see Remark~\ref{rem:alg-precision} for numerical precision).
\end{proposition}

To apply Barvinok's algorithm, we must establish a zero free region for $\Xi$~\cite{BarvinokBook}, which we can obtain rather directly from the positivity proven in Proposition~\ref{prop:scalar}.

\begin{lemma}[Uniform zero-free region]
\label{lem:alg-polydisc}
If $|v_\gamma|\le t^{|\gamma|}w_A(\gamma)$ for every polymer, then
$\Xi(\mathbf v)\ne0$. In particular,
\begin{equation}
 f_{S,B}(q):=\Xi\bigl((v_\gamma^{S,B}q^{|\gamma|})_\gamma\bigr)
 \ne0\qquad (|q|\le\mathcal R).
 \label{eq:alg-univariate}
\end{equation}
Here $f_{S,B}(0)=1$, $f_{S,B}(1)=\mathcal Z(S,B)$, and
$\deg f_{S,B}\le m$.
\end{lemma}
\begin{proof}
$\Xi$ is the independence polynomial of a graph whose vertices are polymers, and we can apply known results for the zero free region of the independence polynomial~\cite{ScottSokal}.
To apply this, set $M_\gamma=t^{|\gamma|}w_A(\gamma)$ and let $\mathbf M = (M_\gamma)_\gamma$ denote the vector storing values of $M_\gamma$ for all $\gamma$. Implication (a)$\Rightarrow$(c) of Theorem~2.10 of~\cite{ScottSokal} says that $\Xi$ is nonzero throughout $|v_\gamma|\le M_\gamma$ if there is a continuous path $\mathbf h(s)$ from $\mathbf h(0)=\mathbf 0$ to $\mathbf h(t)=-\mathbf M$ strictly in the nonpositive real orthant (i.e. no entry of $\mathbf h(s)$ can be positive for $0\le s\le t$) where $\Xi(\mathbf h(s))>0$ is strictly positive in that path.
Proposition~\ref{prop:scalar} provides such a path where
\[
 \mathbf h(s)=(-s^{|\gamma|}w_A(\gamma))_\gamma,\qquad 0\le s\le t,
\]
satisfying $\mathbf h(0)=\mathbf0$, $\mathbf h(t)=-\mathbf M$, and $\Xi(\mathbf h(s))=Q_{G,A}(s)>0$.
For every $(S,B)$ and $|q|\le\mathcal R$,
\[
 |v_\gamma^{S,B}q^{|\gamma|}|
 \le r_\gamma\mathcal R^{|\gamma|}
 =M_\gamma.
\]
Thus $f_{S,B}(q)\ne0$ on the disk $|q|\le\mathcal R$. 
$\deg f_{S,B}\le m$ since any compatible family of polymers is an independent set and there are at most $m$ terms in the Hamiltonian. 
\end{proof}

The zero-free region shown above lets us prove the tail bound in Lemma~\ref{lem:alg-tail}.
\begin{proof}[Proof of Lemma~\ref{lem:alg-tail}]
Fix the current $(S,B)$ and label $v$. The sum in~\eqref{eq:alg-tail}
is the probability of selecting a polymer through $v$ with more than
$K$ labels. Denote the set of these eligible large polymers by
\[
 \mathscr L=\{\gamma:\ v\in\gamma\subseteq S\setminus B,
                 \ G[\gamma]\text{ connected},\ |\gamma|>K\},
 \qquad p_{>K}=\sum_{\gamma\in\mathscr L}p_\gamma.
\]
For a complex number $h$, define a new weight assignment and its
partition function by
\[
 u_\gamma(h)=v_\gamma^{S,B}+h a_\gamma\mathbf1_{\{\gamma\in\mathscr L\}},
 \qquad
 \mathcal Z_h=\Xi(\mathbf u(h))
             =\sum_{\Theta\ \mathrm{compatible}}
                         \prod_{\theta\in\Theta}u_\theta(h).
\]
Every member of $\mathscr L$ contains $v$, so a compatible family
$\Theta$ contains at most one such member. Expanding the displayed
products therefore yields a constant term and a linear term only:
\begin{align}
 \mathcal Z_h
 &=\mathcal Z(S,B)+h\sum_{\gamma\in\mathscr L}a_\gamma
     \sum_{\substack{\Theta\ \mathrm{compatible}\\
             \theta\ \mathrm{compatible\ with}\ \gamma\ (\theta\in\Theta)}}
                      \prod_{\theta\in\Theta}v_\theta^{S,B}\notag\\
 &=\mathcal Z(S,B)+h\sum_{\gamma\in\mathscr L}
                         a_\gamma\mathcal Z(S_\gamma,B_\gamma)
  =\mathcal Z(S,B)(1+h p_{>K}).
 \label{eq:alg-linear-perturbation}
\end{align}
In the inner sum, all polymers incompatible with $\gamma$ are excluded.
The nonzero remaining weights are exactly those on
$S_\gamma=S\setminus N_G[\gamma]$, with forbidden labels
$B_\gamma=B\cap S_\gamma$; all their weights are unchanged.
This proves the second equality. The last equality is the conditional
probability formula~\eqref{eq:alg-pgamma}.

If $|h|\le(\mathcal R^{K+1}-1)/2$, every changed weight obeys
\[
 |u_\gamma(h)|\le r_\gamma+2|h|r_\gamma
 \le\mathcal R^{|\gamma|}r_\gamma=t^{|\gamma|}w_A(\gamma).
\]
Unchanged weights obey the same bound, so Lemma~\ref{lem:alg-polydisc}
excludes a zero of $\mathcal Z_h$ in this disk. If $p_{>K}>0$, the
zero $h=-1/p_{>K}$ of~\eqref{eq:alg-linear-perturbation} must lie
outside the disk. This gives~\eqref{eq:alg-tail}.
\end{proof}

Now we turn to estimating $f_{S,B}(q)$ in~\eqref{eq:alg-univariate}. We will combine Barvinok's error bound in the Taylor truncation~\cite[Lemma~2.2.1]{BarvinokBook} with the polymer enumeration and runtimes from~\cite[Theorem~2.2]{HPR}. Implementations of this are developed in~\cite{PatelRegts,YYZ}.

\Needspace{16\baselineskip}
\begin{lemma}[Approximating the conditional partition function]
\label{lem:alg-counting}
Fix $(S,B)$ and denote the Taylor series of $\log f_{S,B}$ as
\begin{equation}
 \log f_{S,B}(q)=\sum_{j\ge1}c_jq^j,
 \label{eq:alg-generic-polymer}
\end{equation}
using the logarithm with value zero at $q=0$.
If a weight $v_\gamma^{S,B}$ on $k$ labels can be evaluated to absolute error $2^{-p}$ in $2^{O(k)}\operatorname{poly}(n,m,k,p)$ time, then the Taylor approximation $\exp(\sum_{j=1}^L c_j)$ can be evaluated in $\operatorname{poly}(n,m,L)(\Delta+1)^{O(L)}$ time
(with numerical errors as in Remark~\ref{rem:alg-precision}).
Its truncation error satisfies
\begin{equation}
 \left|\log\mathcal Z(S,B)-\sum_{j=1}^L c_j\right|
 \le\frac{m\mathcal R^{-(L+1)}}{(L+1)(1-\mathcal R^{-1})}.
 \label{eq:alg-log-tail}
\end{equation}
Thus an error at most $\xi$ in~\eqref{eq:alg-log-tail} gives an
$e^{\pm\xi}$ multiplicative approximation before numerical error.
\end{lemma}
\begin{proof}
We split the proof into that for the error from the Taylor series truncation and the runtime.

\emph{Taylor error.}
By Lemma~\ref{lem:alg-polydisc}, the polynomial $f_{S,B}$ has degree $d\le m$ and zeros $\zeta_1,\ldots,\zeta_d$ outside the disk $|q|\le\mathcal R$. Following the proof of \cite[Lemma~2.2.1]{BarvinokBook}, factorization over zeros of $f_{S,B}$ gives
\begin{equation}
 f_{S,B}(q)=\prod_{\ell=1}^d(1-q/\zeta_\ell),\qquad
 \log f_{S,B}(q)=-\sum_{j\ge1}\frac{q^j}{j}
                                   \sum_{\ell=1}^d\zeta_\ell^{-j}.
 \label{eq:alg-root-expansion}
\end{equation}
Hence $|c_j|\le m/(j\mathcal R^j)$, and the tail at $q=1$ is at most
\[
 \sum_{j>L}\frac{m}{j\mathcal R^j}
 \le\frac{m\mathcal R^{-(L+1)}}{(L+1)(1-\mathcal R^{-1})}.
\]
The polynomial is real and nonzero on $[0,1]$ and starts at one, so the logarithm there is real. Thus
\begin{equation}
 \frac{\exp(\sum_{j=1}^L c_j)}{\mathcal Z(S,B)}
 =\exp\!\left(\sum_{j=1}^L c_j-\log\mathcal Z(S,B)\right),
 \label{eq:alg-exponentiate}
\end{equation}
which gives the multiplicative estimate. 

\emph{Coefficient computation.}
The coefficients $c_1, \dots, c_L$ can be computed using the methods in~\cite{PatelRegts,YYZ} (see also~\cite{MM,PZC}).
We group the coefficients of derivatives of $\log f_{S,B}$ according to the set of term labels they use.
Polymers are connected sets of $G[S]$, and there are at most $m(\Delta+1)^{O(L)}$ connected sets of $G[S]$ of size at most $L$. These are enumerable within that bound up to polynomial factors in $n$ and $m$ (see Lemma~2.4 of~\cite{HPR}). This enumeration also appears in~\cite[Lemma~3.4]{PatelRegts}.
The coefficient decomposition in~\cite{PatelRegts,YYZ} then gives
\begin{equation}
 f_U(q)=\sum_{W\subseteq U}q^{|W|}
              \prod_{\gamma\in\comp(G[W])}v_\gamma^{S,B},
 \qquad U\subseteq S.
 \label{eq:alg-connected-coeff}
\end{equation}
Throughout this calculation $(S,B)$ and its weights are fixed and $f_S=f_{S,B}$. The coefficient formula derived by applying inclusion--exclusion is~\cite[Lemma 17]{YYZ}
\begin{equation}
 c_j=\sum_{\substack{T\subseteq S,\ G[T]\ \mathrm{connected}\\1\le|T|\le j}}
       \ \sum_{U\subseteq T}(-1)^{|T|-|U|}[q^j]\log f_U(q).
 \label{eq:alg-coefficient-sum}
\end{equation}
The convenient fact about~\eqref{eq:alg-coefficient-sum} is that it computes $c_j$ using connected sets of size at most $j$.
In fact, $h_j(T)=\sum_{U\subseteq T}(-1)^{|T|-|U|}[q^j]\log f_U(q)$ vanishes unless $T$ is connected and has at most $j$ vertices.
To see why, let $G[T]$ be disconnected, and write $T=T_1\sqcup T_2$ with both parts nonempty and no edge between them. For every $U\subseteq T$, the component weights in \eqref{eq:alg-connected-coeff} factor, giving
\[
 f_U(q)=f_{U\cap T_1}(q)f_{U\cap T_2}(q).
\]
Thus $\log f_U=\log f_{U\cap T_1}+\log f_{U\cap T_2}$.
In the alternating sum $h_j(T)=\sum_{U\subseteq T}(-1)^{|T|-|U|}[q^j](\log f_{U\cap T_1}+\log f_{U\cap T_2})$, the first term cancels upon summing over subsets of $T_2$, and the second cancels upon summing over subsets of $T_1$. Hence $h_j(T)=0$. 
Next, expand the logarithm as
\[
 \log f_U(q)=\sum_{r\ge1}\frac{(-1)^{r-1}}{r}
                              (f_U(q)-1)^r.
\]
A contribution to its coefficient of $q^j$ uses nonempty sets $W_1,\ldots,W_r\subseteq U$ with $\sum_{\ell=1}^r|W_\ell|=j$.  Because the weights are fixed, this contribution occurs with the same coefficient for every $U$ containing $W$.
Its contribution to $h_j(T)=\sum_{U\subseteq T}(-1)^{|T|-|U|}[q^j]\log f_U(q)$ is consequently multiplied by
\[
 \sum_{U:\,W\subseteq U\subseteq T}(-1)^{|T|-|U|}
 =
 \begin{cases}
  1,&W=T,\\
  0,&W\ne T.
 \end{cases}
\]
If $|T|>j$, then $W\ne T$ for every such contribution, proving
$h_j(T)=0$.

For each $T$, the subset sums in~\eqref{eq:alg-coefficient-sum} cost $2^{O(L)}$ operations to compute up to polynomial factors since each $f_U(q)$ is a sum of $2^{O(L)}$ monomials.
Summing over the enumerated connected sets proves the claimed runtime of $\operatorname{poly}(n,m,L)(\Delta+1)^{O(L)}$.
\end{proof}

We now put the pieces together to prove Proposition~\ref{prop:alg-estimateZ}.
\begin{proof}[Proof of Proposition~\ref{prop:alg-estimateZ}]
Apply Lemma~\ref{lem:alg-counting} with
\[
 L=\left\lceil\frac4\delta\log\frac{16(m+1)}{\delta\xi}\right\rceil.
\]
We need to account for the runtime to compute the weights $v_\gamma^{S,B}$ in~\eqref{eq:alg-mixed} which are either set to $-r_\gamma=-\sigma^{|\gamma|}w_A(\gamma)$ or $\tau(D_\gamma)$.
For $k=|\gamma|$, $-r_\gamma$ can be computed in $\operatorname{poly}(k,p)$ time at precision $p$ by computing the connected components of $A[\gamma]$ giving $w_A(\gamma)$.
Products of the $k$ Pauli generators span an algebra with at most $2^k$ Pauli basis strings. Computing the subset exponentials and inclusion--exclusion in~\eqref{eq:FD} evaluates $\tau(D_\gamma)$ to absolute error $2^{-p}$ in $2^{O(k)}\operatorname{poly}(n,k,p)$ time. Section~\ref{sec:alg-states} gives the details. 

For the error estimate, the temperature slack implies
\begin{equation}
 \log\mathcal R\ge\frac\delta4,\qquad
 1-\mathcal R^{-1}\ge\frac\delta4,\qquad
 \mathcal R-1\ge\frac\delta4.
 \label{eq:alg-slack}
\end{equation}
Indeed, $z\le z_2<1/2$ and $\tanh'\ge1/2$ on $[0,z_2]$, so
$t-\tanh(\beta J)\ge\delta z/2\ge\delta t/2$.
Thus $\mathcal R^{-1}=\sigma/t\le1-\delta/4$, giving all three
inequalities. With the choice of $L$,
\[
 \mathcal R^{-(L+1)}\le e^{-\delta L/4}
 \le\frac{\delta\xi}{16(m+1)},
\]
so the error in~\eqref{eq:alg-log-tail} is at most $\xi/4$.
Work with enough precision so that computed value $\sum_{j=1}^L c_j$ has additive error at most $\xi/4$, as justified in Remark~\ref{rem:alg-precision}. Exponentiation contributes a further multiplicative factor $e^{\pm\xi/4}$.
The total factor is therefore $e^{\pm3\xi/4}$, proving \eqref{eq:alg-Z-contract}. Lemma~\ref{lem:alg-counting} gives the runtime.

\end{proof}

\subsection{Sampling product states from the polymers}
\label{sec:alg-states}

It remains to sample $\rho_\Gamma$ in~\eqref{eq:alg-state-decomposition} which as a reminder is equal to
\[
\rho_\Gamma =
 \left(\bigotimes_{\gamma\in\Gamma}
       \frac{R_\gamma}{\Tr_{U_\gamma}R_\gamma}\right)
       \otimes\frac{I_{U_\Gamma^c}}{2^{|U_\Gamma^c|}}.
\]
Each $\rho_\Gamma$ is a product of terms of the form $R_\gamma/\Tr_{U_\gamma}R_\gamma$ which act on disjoint sets of qubits so we treat them independently.
Identity operators acting on the remaining qubits can be handled by sampling uniformly at random from Pauli Z eigenstates.

\begin{proposition}[Local factor sampling]
\label{prop:alg-factor}
For a polymer of size $k$ and $0<\xi<1$, one can compute $a_\gamma$ to relative error $\xi$ and sample the product-state mixture to total variation error $\xi$ in $2^{O(k)}\operatorname{poly}(n,k,\log(1/\xi))$ time, with
$\Delta,\delta$ fixed and numerical conventions as in
Remark~\ref{rem:alg-precision}.
\end{proposition}
\begin{proof}
A polymer of size $k$ contains $k$ Pauli terms spanning an algebra of at most $2^k$ Pauli strings.
As a reminder, $R_\gamma = r_\gamma I+D_\gamma$.
We can write $D_\gamma$ in the Pauli basis of this algebra so $D_\gamma=d_I I+\sum_{P\ne I}d_PP$, so $\tau(D_\gamma)=d_I$.
The runtime needed to obtain this expansion up to Pauli coefficient $\ell^1$ error at most $2^{-p}$ is $2^{O(k)}\operatorname{poly}(n,k,p)$ by computing the exponentials and inclusion--exclusion arguments within this algebra.
The decomposition in Lemma~\ref{lem:cone} gives
\begin{equation}
 \begin{split}
 R_\gamma&=c_\gamma I+
      \sum_{P\ne I}|d_P|\bigl(I+\operatorname{sgn}(d_P)P\bigr),\\
 c_\gamma&=r_\gamma+d_I-\sum_{P\ne I}|d_P|
       \ge(\mathcal R-1)r_\gamma>0,\qquad
 a_\gamma=c_\gamma+\sum_{P\ne I}|d_P|.
 \end{split}
 \label{eq:alg-local-mixture}
\end{equation}
Division of the coefficients by $a_\gamma$ gives probability $|d_P|/a_\gamma$ for a Pauli $P$. We sample an eigenstate by first sampling a Pauli and then a uniformly random product eigenstate of that Pauli. This is detailed in Algorithm~\ref{alg:sample-factor}. 

\begin{algorithm}[H]
\caption{$\operatorname{SampleFactor}(\gamma,\xi)$}
\label{alg:sample-factor}
\small
\begin{algorithmic}[1]
\Require A selected polymer $\gamma$ and total variation tolerance $\xi$.
\State Compute $D_\gamma=d_I I+\sum_{P\ne I}d_PP$ and $r_\gamma,c_\gamma,a_\gamma$.
\State Draw the identity branch with probability $c_\gamma/a_\gamma$, or a
\Statex \hspace{\algorithmicindent}nonidentity Pauli branch $P$ with probability $|d_P|/a_\gamma$.
\If{the identity branch is chosen}
 \State \Return independent uniformly random $Z$ eigenstates on $U_\gamma$.
\EndIf
\State List the nonidentity sites of $P$ as $i_1,\ldots,i_\ell$.
\State Draw independent uniform signs $\lambda_1,\ldots,\lambda_{\ell-1}\in\{-1,1\}$.
\State $\lambda_\ell\gets\operatorname{sgn}(d_P)\prod_{j<\ell}\lambda_j$.
\State At $i_j$, output the eigenstate of $P_{i_j}$ with eigenvalue $\lambda_j$;
\Statex \hspace{\algorithmicindent}on $U_\gamma\setminus\supp(P)$, use independent uniform $Z$ eigenstates.
\State \Return the tensor product of these single-qubit states.
\end{algorithmic}
\end{algorithm}

For numerical stability, it suffices to approximate to error
\[
 E:=\normp{\widetilde D_\gamma-D_\gamma}
      +|\widetilde r_\gamma-r_\gamma|
 \le\frac{\xi}{16}(\mathcal R-1)r_\gamma.
\]
The errors then in $c_\gamma$ and $a_\gamma$ are at most $E$, so both remain positive and $a_\gamma$ is estimated to relative error at most $\xi/16$.
The same bound changes the normalized signed-Pauli distribution by at most $2E/a_\gamma\le\xi/8$ in total variation.
Expanded in the Pauli basis, $D_\gamma$ is a weighted sum of at most $2^k$ Pauli monomials.
Since $\sigma\ge t/2$, $\eta\ge1/3$, and $\alpha\ge1$, we have $r_\gamma\ge(t/6)^k$. Consequently $p=O_{\Delta,\delta}(k+\log(1/\xi))$ bits of precision suffice for bounding the numerical error in the coefficients of $D_\gamma$ and $r_\gamma$.

\end{proof}

\begin{remark}[Numerical input and bit complexity]
\label{rem:alg-precision}
To describe numerical complexity, assume the coefficients of the Pauli Hamiltonian, $\beta$, and $\epsilon$ are given as binary rational numbers, and let $\mathsf L$ be their total input length. The displayed runtimes suppress a polynomial dependence on $\mathsf L$. The full sampler has complexity
\[
 \operatorname{poly}(n,m,\mathsf L,K,L)(\Delta+1)^{O(K+L)}.
\]

Only polynomially many precision bits are needed. The estimates needed for sampling are controlled by the positive margin in Proposition~\ref{prop:alg-factor}. For the partition function oracle, one can round each polymer weight to a fraction in $[-1,1]$ and perform the necessary calculations with rational arithmetic. Since $f_U(0)=1$ and $|v_\gamma^{S,B}|<1$, errors of $2^{-p}$ in the precision of the weights change $\sum_{j=1}^L c_j$ in Lemma~\ref{lem:alg-counting} by at most $m\exp(O_\Delta(L^2))2^{-p}$. Thus $p=O_\Delta(L^2+\log(m/\xi))$ suffices for additive error $\xi/4$.
The bounds on the roots in~\eqref{eq:alg-root-expansion} gives
$|\log\mathcal Z(S,B)|=O_\delta(m)$, so even an exponentially small
$\mathcal Z(S,B)$ has a relative approximation with polynomial bit
length.

Finally, round each polymer-decision distribution to probabilities with a common power-of-two denominator, with total variation error at most $\epsilon/(4m)$ per decision. Rounding within Algorithm~\ref{alg:sample-factor} is included in its tolerance $\xi$.
\end{remark}

\subsection{Proof of the sampling theorem}
\label{sec:alg-conclusion}

\begin{proof}[Proof of Theorem~\ref{thm:sampling}]
For $\rho_\beta(H)=\sum_\Gamma \pi(\Gamma)\rho_\Gamma$, the algorithm first samples a set of polymers according to $\pi$ and then samples a product state from $\rho_\Gamma$. 
To sample a set of polymers, run Algorithm~\ref{alg:product-sampler}. For each selected polymer, run Algorithm~\ref{alg:sample-factor} to sample from $\rho_\Gamma$ filling uncovered qubits there with independent uniform $Z$ eigenstates.

Without the size cutoff and with exact probabilities, the conditional probabilities \eqref{eq:alg-pgamma} sample $\pi$. 
Conditioning the decision to sample on not selecting a polymer of size larger than $K=O\left(\log\frac{m}{\epsilon}\right)$ changes the distribution by at most $\epsilon/(8m)$ in total variation, by Lemma~\ref{lem:alg-tail}. Following \eqref{eq:alg-continuation-weights} and \eqref{eq:alg-draw}, $W_y$ denotes the weights over eligible polymers of size no larger than $K$ with distribution $q_y=W_y/(\sum_{y'}W_{y'})$.
Setting $\xi=\epsilon/(100m)$, Proposition~\ref{prop:alg-estimateZ} constructs estimate $\widehat W_y$ to $W_y$ up to multiplicative error
\[
 (1-\xi)e^{-\xi}\le\frac{\widehat W_y}{W_y}
                  \le(1+\xi)e^\xi.
\]
With this error bound, we also have $\sum_y|\widehat W_y-W_y| \le 3\xi\sum_y W_y$.
We thus have
\[
 \frac12\sum_y\left|
    \frac{\widehat W_y}{\sum_{y'}\widehat W_{y'}}-
    \frac{W_y}{\sum_{y'}W_{y'}}\right|
 \le\frac{\sum_y|\widehat W_y-W_y|}{\sum_yW_y} \le 3\xi.
\]
Thus numerical evaluation changes the law of $W_y$ by at most $3\xi$.

We now accumulate these errors over the at most $m$ choices of the sampler (each polymer decision removes at least one undecided label, so there are at most $m$ decisions). 
Let $\mu_i$ and $\widehat\mu_i$ be the distributions of the first $i$ decisions in the ideal sampler and in the sampler using the cutoff and estimated continuation weights, respectively. Write $p_i(\,\cdot\mid h)$ and $\widehat p_i(\,\cdot\mid h)$ for the distribution of the next decision given the history $h$. The preceding bounds give, uniformly over
valid histories,
\[
 d_{\mathrm{TV}}\bigl(p_i(\,\cdot\mid h),
                      \widehat p_i(\,\cdot\mid h)\bigr)
 \le \frac{\epsilon}{8m}+3\xi.
\]
Summing over all possible histories,
\begin{align*}
 d_{\mathrm{TV}}(\mu_i,\widehat\mu_i)
 &=\frac12\sum_{h,y}
   \left|\mu_{i-1}(h)p_i(y\mid h)
   -\widehat\mu_{i-1}(h)\widehat p_i(y\mid h)\right|\\
 &\le d_{\mathrm{TV}}(\mu_{i-1},\widehat\mu_{i-1})
   +\sum_h\widehat\mu_{i-1}(h)\,
      d_{\mathrm{TV}}\bigl(p_i(\,\cdot\mid h),
                           \widehat p_i(\,\cdot\mid h)\bigr)\\
 &\le d_{\mathrm{TV}}(\mu_{i-1},\widehat\mu_{i-1})
   +\frac{\epsilon}{8m}+3\xi.
\end{align*}
The first inequality uses that each conditional distribution sums to one. Since $\mu_0=\widehat\mu_0$, iterating this inequality yields
\[
 d_{\mathrm{TV}}(\pi,\widehat\pi)
 \le d_{\mathrm{TV}}(\mu_m,\widehat\mu_m)
 \le m\left(\frac{\epsilon}{8m}+3\xi\right),
\]
where $\widehat\pi$ is the resulting law of the sampler for $\Gamma$. 
Now fix a family $\Gamma$ outputted by this approximate sampler.
Let $\nu_\Gamma$ and $\widehat\nu_\Gamma$ be the exact and implemented laws of the product-state output conditioned on selecting $\Gamma$. Applying the triangle inequality gives
\[
 d_{\mathrm{TV}}(\nu_\Gamma,\widehat\nu_\Gamma)
 \le \sum_{\gamma\in\Gamma}
       d_{\mathrm{TV}}(\nu_\gamma,\widehat\nu_\gamma)
 \le |\Gamma|\xi\le m\xi,
\]
by Proposition~\ref{prop:alg-factor}. 
The ideal and approximated laws are therefore bounded in TV distance by
\begin{align*}
 d_{\mathrm{TV}}\left(
     \sum_\Gamma\pi(\Gamma)\nu_\Gamma,\,
     \sum_\Gamma\widehat\pi(\Gamma)\widehat\nu_\Gamma
                 \right)
 &\le d_{\mathrm{TV}}(\pi,\widehat\pi)
   +\sum_\Gamma\widehat\pi(\Gamma)
       d_{\mathrm{TV}}(\nu_\Gamma,\widehat\nu_\Gamma)\\
 &\le m\left(\frac{\epsilon}{8m}+4\xi\right).
\end{align*}
For the first inequality, insert $\sum_\Gamma\widehat\pi(\Gamma)\nu_\Gamma$ and use the triangle
inequality.

Finally, to account for finite precision, we allocate $\epsilon/(4m)$ additional total variation error to rounding weights as stated in Remark~\ref{rem:alg-precision}. This adds at most $\epsilon/4$ to the output error. Thus, if $p_\psi$ and $q_\psi$ are the ideal and final implemented output probabilities,
\[
 d_{\mathrm{TV}}(p,q)
 \le m\left(\frac{\epsilon}{8m}+4\xi\right)+\frac{\epsilon}{4}
 <\epsilon.
\]
The ideal law $p_\psi$ averages to $\rho_\beta(H)$ and since each $\ket\psi\bra\psi$ has trace norm one,
\[
 \frac12\left\|
      \sum_\psi q_\psi\ket\psi\bra\psi-\rho_\beta(H)
             \right\|_1
 =\frac12\left\|
      \sum_\psi(q_\psi-p_\psi)\ket\psi\bra\psi
             \right\|_1
 \le\frac12\sum_\psi|q_\psi-p_\psi|
 <\epsilon.
\]
This proves~\eqref{eq:sampling-guarantee}.

Across all decisions there are at most $m(\Delta+1)^{O(K)}$ candidates (either polymers or vertices) for sampling, counting repetitions. Setting $K=O\left(\delta^{-1}\log\frac{m+1}{\delta\epsilon}\right)$ and similarly $ L=O\left(\delta^{-1}\log\frac{m+1}{\delta\epsilon}\right)$, Propositions~\ref{prop:alg-estimateZ} and~\ref{prop:alg-factor} give total runtime
\begin{equation}
 \operatorname{poly}(n,m,K,L)(\Delta+1)^{O(K+L)}.
 \label{eq:alg-runtime}
\end{equation}
With the numerical convention in Remark~\ref{rem:alg-precision}, this
is polynomial in $n,m,1/\epsilon$ for fixed $\Delta,\delta$.
\end{proof}

\section{Discussion}
\label{sec:discussion}

The universal threshold in Theorem~\ref{thm:fixed} is shown to be optimal and approached by a family of commuting Pauli Hamiltonians whose term-overlap graphs are trees. This picture bears similarities to several classical arguments in which a recursive description on trees determines a given threshold. Examples include Shearer positivity criteria for the independence polynomial, walk tree recursions for approximate counting, and proofs of correlation decay and zero freeness~\cite{ScottSokal,Weitz,LSSContraction,ShaoSun,PetersRegtsIsing,FernandezProcacci}.
Trees do not always constitute a class of systems approaching a statistical threshold. For example, Sly constructed a multistate spin system with uniqueness on the $d$-regular tree but nonuniqueness on another graph of the same degree~\cite{SlyTreesVsGraphs}. 

Recursively calculated thresholds appear in other areas of quantum information.
For quantum satisfiability, the quantum Lov\'asz local lemma and the quantum Shearer bound give graph-dependent conditions for frustration-freeness~\cite{AKS,SattathFF}. Shearer's bound is tight in the general setting where local dimensions are allowed to grow~\cite{HeLiSunZhangQLLL}. Another example of a quantum threshold appears in magic-state distillation, where the success of a purification protocol is governed by a recurrence formula~\cite{BravyiKitaevMagic}.

Our proof establishes positivity after collecting the full expansion, rather than preserving positivity through the pinning steps of~\cite{BLMT,PZC}. At bounded degree, the resulting decomposition supports the approximate sampler in Theorem~\ref{thm:sampling}. Its conditional polymer partition functions remain in a common zero-free region, which permits their efficient approximation. Our sampler does not by itself extend the temperature guarantees of the earlier pinning procedures, and it is unclear if the pinning procedures give efficient algorithms up to the threshold $z_\Delta$.

Several open questions remain. The sharpness construction allows the locality of a term to grow with $\Delta$, so the optimal constants at a fixed locality are undetermined (e.g. say $k=2$ local Hamiltonians). One can also ask for thermal thresholds relative to larger classes than product states, such as convex mixtures of states prepared by shallow circuits or circuits in a fixed level of the magic hierarchy~\cite{AG,Parham}. Finally, random Hamiltonians may have typical thresholds different from the worst-case bound. Results on product-state and ground-state energies of quantum $p$-spin models~\cite{AGK} provide one setting in which to investigate this distinction. 

More broadly, we anticipate there are connections between the proof ideas here and the growing literature on quantum or classical algorithms for Gibbs state preparation~\cite{CKG,CKBG}. Mixing time bounds have been proven for suitable Gibbs samplers on commuting models~\cite{KastoryanoBrandao,KochanowskiEtAl}. Recent work on a quantum Dobrushin condition~\cite{QuantumDobrushin} and mixing time bounds from cluster expansions~\cite{Bergamaschi} or Lieb-Robinson proofs~\cite{RFAOptimal} give guarantees of fast mixing times at high-temperature but the thresholds there are likely not optimally determined. Another variant of this question asks when it is computationally challenging to sample from the computational basis of an entangled Gibbs state~\cite{bergamaschi2024quantum}. 
Finding exact thresholds for zeros in the physical partition function (Fisher zeros) of Pauli Hamiltonians is also an interesting question where substantial progress has been made in~\cite{HMS,ZK,PZC,MM,MannWaite,MH}.
One should not immediately expect that thresholds of separability, zero-freeness of the physical partition function, and rapid mixing to coincide automatically~\cite{ZK,PZC}. For example, Remark~\ref{rem:zeros} gives an explicit example of a separation between separability and zero-free thresholds for the Hamiltonian family studied here.
Nonetheless, it would be interesting to see if the ideas here can help in improving bounds on these other thresholds.

\section*{Acknowledgements}
I thank Eric Anschuetz, Saeed Mehraban, and Alexander Zlokapa for helpful discussions. 
I acknowledge support from the Hastings Initiative for AI and Humanity and NSF Award 2624545. 

\medskip
\noindent
\textbf{Personal note:} 
My goal in posting this is to spark new conversations and share with the community a set of results that I found very enlightening.
I hope we find time to collaborate and talk to each other about the interesting things we discover. 
Otherwise, this manuscript will only be another addition to the recent proliferation of papers in our community.
So in that light, if this paper interests you, I would love to hear from you.
}

\bibliographystyle{quantum}
\bibliography{references}

\end{document}